\documentclass[11pt,onecolumn]{article}

\usepackage{amsmath,amsthm,amssymb,amsfonts,epsfig,graphicx}
\usepackage[linesnumbered,ruled,vlined]{algorithm2e}
\usepackage[left=0.9in,top=0.9in,right=0.9in,bottom=0.5in,nohead,textheight=10in,footskip=0.3in]{geometry}

\newcommand{\eps}{\varepsilon}
\newtheorem*{rep@theorem}{\rep@title}
\newcommand{\newreptheorem}[2]{%
\newenvironment{rep#1}[1]
{ \def\rep@title{#2 \ref{##1}} \begin{rep@theorem}} {\end{rep@theorem}} }

\newtheorem{theorem}{Theorem}
\newreptheorem{theorem}{Theorem}
\newtheorem{lemma}[theorem]{Lemma}

\newtheorem{observation}[theorem]{Observation}
\newtheorem{proposition}[theorem]{Proposition}
\newtheorem{corollary}[theorem]{Corollary}

\theoremstyle{definition}

\newcommand{\bE}{{\mathbb E}}
\newcommand{\ee}{\mathrm{e}}
\newcommand{\roff}{P_{\text{offset}}}

\DeclareMathOperator{\argmin}{argmin}
\DeclareMathOperator{\lb}{LB}

\newcommand{\cleft}{\mathcal C^*_{k,\text{left}}}
\newcommand{\cleftl}{\mathcal C^*_{\ell,\text{left}}}
\newcommand{\rleft}{R}
\newcommand{\sleft}{S}

\newcommand{\jtrunc}{t}

\newcommand{\xmat}{\mathbf{X}}
\newcommand{\lmax}{\tau}
\newcommand{\piValue}{3.9}
\newcommand{\thetaValue}{0.555}
\newcommand{\lmaxValue}{0.604}
\newcommand{\kappaValue}{0.744}
\newcommand{\rhoValue}{1.398}
\newcommand{\gammaValue}{0.00547}
\newcommand{\betaValue}{1.935}
\newcommand{\deltaValue}{1.3}
\newcommand{\cZeroValue}{0.17556}
\newcommand{\cOneValue}{0.591909}
\newcommand{\cTwoValue}{0.33965}
\newcommand{\cThreeValue}{0.8277532}
\newcommand{\cFourValue}{0.035898}
\newcommand{\cFiveValue}{0.004562}
\newcommand{\cSixValue}{0.0676}
\begin{document}

\title{Dependent rounding with strong negative-correlation,  and scheduling on unrelated machines to minimize completion time\footnote{This is an extended version of a paper appearing in the 2024 annual ACM-SIAM Symposium on Discrete Algorithms (SODA). It includes more details about the numerical analysis and  slightly improved computation of the approximation ratio.}}
\author{David G. Harris\footnote{University of Maryland. Email: {\tt davidgharris29@gmail.com}}}

\maketitle

\abstract{We describe a new dependent-rounding algorithmic framework for bipartite graphs. Given a fractional assignment $\vec x$ of values to edges of a graph $G = (U \cup V, E)$, the algorithms return an integral solution $\vec X$ such that each right-node $v \in V$ has at most one neighboring edge $f$ with $X_f = 1$, and the variables $X_e$ also satisfy broad nonpositive-correlation properties. In particular, for any edges $e_1, e_2$ sharing a left-node $u \in U$, the variables $X_{e_1}, X_{e_2}$ have strong negative correlation, i.e. the expectation of $X_{e_1} X_{e_2}$ is significantly below $x_{e_1} x_{e_2}$. 

This algorithm is based on generating negatively-correlated Exponential random variables and using them for a rounding method inspired by a  contention-resolution scheme of Im \& Shadloo (2020). Our algorithm gives stronger and much more flexible negative correlation properties.

Dependent rounding schemes with negative correlation properties have been used for approximation algorithms for job-scheduling on unrelated machines to minimize weighted completion times (Bansal, Srinivasan, \& Svensson (2021),  Im \& Shadloo (2020), Im \& Li (2023)). Using our new dependent-rounding algorithm, among other improvements, we obtain a $1.398$-approximation for this problem. This significantly improves over the prior $1.45$-approximation ratio of Im \& Li (2023).
}

\bigskip

\section{Introduction}
Many discrete optimization algorithms are based on the following framework: we start by solving some relaxation of the problem instance (e.g., a linear program), obtaining a fractional solution $\vec x = (x_1, \dots, x_n)$. We then want to convert it into an integral solution $\vec X = (X_1, \dots, X_n)$, with probabilistic properties related to $\vec x$, while also satisfying any needed hard combinatorial constraints for the problem. This general framework often goes by the name \emph{dependent rounding}, since in general the combinatorial constraints will necessarily induce dependencies among the variables $X_i$. 

Some forms of dependent rounding are highly tailored to specific algorithmic problems, while others are very general. These frequently appear in clustering problems, for instance, where there is a hard constraint that every data item must be mapped to a cluster-center; see for example \cite{cluster1, cluster2}. It also appears in a number of job-scheduling problems \cite{gandhi}, where there is a hard constraint that every job must be assigned to some processing node.

One particularly important property is that the variables $X_i$ should satisfy \emph{nonpositive-correlation} properties; namely, for certain subsets $L \subseteq \{1, \dots, n \}$, we should have
\begin{equation}
\label{negeq1}
\bE \Bigl[ \prod_{i \in L} X_i \Bigr] \leq \prod_{i \in L} \bE[X_i] = \prod_{i \in L} x_i
\end{equation}

For example, this property leads to Chernoff-type concentration bounds. Depending on the combinatorial constraints, such an inequality may only be satisfied for certain limited choices of subset $L$. Note that independent rounding would certainly satisfy this property for all sets $L$.

For some combinatorial problems, nonpositive-correlation is not enough; we need \emph{strong negative correlation}. Namely, for certain edge sets $L$, we want a bound of the form
\begin{equation}
\label{negeq2}
\bE \Bigl[ \prod_{i \in L} X_i \Bigr] \leq (1 - \phi) \cdot \prod_{i \in L} x_i
\end{equation}
for some parameter $\phi \gg 0$; the precise value of $\phi$ may depend on the set $L$ and/or the values $x_i$. 
 
One particularly interesting application comes from \emph{job scheduling on unrelated machines to minimize completion time}. We will discuss this problem in more detail later, but for now let us provide a brief summary. We have a set of machines $\mathcal M$ and a set of jobs $\mathcal J$, where each job $j \in \mathcal J$ has a given processing time on each machine $i \in \mathcal M$. We want to minimize the objective function $\sum_{j \in \mathcal J} w_j C_j$, where $C_j$ is the completion time of $j$ on its chosen machine and $w_j$ is a weight function. A long series of  approximation algorithms have been developed for this problem, often based on sophisticated convex relaxations and dependent-rounding algorithms.

The aim of this paper is to obtain a more unified and general picture of  dependent rounding with strong negative correlation. Our  algorithm, along with some other improvements, leads to a $\rhoValue$-approximation algorithm for the above machine-scheduling problem, which improves over the algorithm of \cite{shi2} while also being simpler and more generic. 

\subsection{Definitions and results for the bipartite-rounding setting}
\label{def-sec}
Many combinatorial assignment problems, such as machine-scheduling, can be abstracted in the following framework: we consider a bipartite graph $G = (U \cup V, E)$, where the fractional variables $\vec x \in [0,1]^E$ are associated to its edges. Our goal is to generate rounded random variables $\vec X \in \{0,1 \}^{E}$ which match the fractional variables, while also achieving nonpositive-correlation or strong negative correlation among certain subset of values $X_e$ corresponding to the vertices.

We call $U$ and $V$ left-nodes and right-nodes respectively. For any vertex $w \in U \cup V$ we define $\Gamma(w)$ to be the edges incident to $w$. 

The \emph{bipartite rounding algorithm} of \cite{gandhi} is a powerful example of this framework. This algorithm, which has found numerous applications in optimization problems, guaranteed a limited case of Eq.~(\ref{negeq1}): namely, it held for any edge set $L$ which was a subset of the neighborhood of any vertex.

Our first main contribution, in Sections~\ref{round-sec1} and~\ref{round-sec}, is to develop a new bipartite-graph rounding algorithm based on negatively correlated Exponential random variables. The construction has two stages. First, we give a general method of generating Exponential random variables with  certain types of negative correlation. Second, we use this in a graph contention-resolution scheme, which takes as input a ``rate'' vector $\vec \rho \in [0,1]^{E}$; we will discuss its role shortly.

Let us say that a set of edges $L \subseteq E$ is \emph{stable} if it does not contain any pair of edges of the form $e = (u,v), e' = (u',v') $, where $(u, v')$ is also an edge in $E$. Equivalently, $L$ does not contain any pair of edges whose distance is precisely two in the line graph of $G$.

We get the following result:

\begin{theorem}
\label{thm1}
Suppose  $\sum_{e \in \Gamma(v)} x_v \leq 1$ for all right-nodes $v \in V$, and $\sum_{e \in \Gamma(u)} \rho_u \leq 1$ for all left-nodes $u \in U$. Then the algorithm \textsc{DepRound} satisfies the following properties (A1) --- (A4):
\begin{itemize}
\item[(A1)] For each right-node $v \in V$, we have $\sum_{e \in \Gamma(v)} X_e \leq 1$ with probability one.
\item[(A2)] For each edge $e$ there holds $\bE[X_e] = x_e$
\item[(A3)] For any stable edge-set $L \subseteq E$, there holds $\bE[ \prod_{e \in L} X_e ] \leq \prod_{e \in L} x_e$
\item[(A4)] For any pair of edges $e_1 = (u,v_1), e_2 = (u,v_2)$ sharing a common left-node, we have $$
\bE[X_{e_1} X_{e_2}]  \leq (1 - \Phi(x_{e_1}, x_{e_2}; \rho_{e_1}, \rho_{e_2})) x_{e_1} x_{e_2},
$$ where we define the function $\Phi(x_1, x_2; \rho_1, \rho_2): [0,1]^4 \rightarrow [0,1]$  by
$$
\Phi(x_1, x_2; \rho_1, \rho_2) =    \begin{cases}
\frac{ ( (1-\rho_{1})^{1 - 1/x_{1}} - 1) ( (1-\rho_{2})^{1 - 1/x_{2}} - 1)}{ (1 -\rho_{1})^{1-1/x_{1}}
  (1-\rho_{2})^{1-1/x_{2}}+\rho_{1}+\rho_{2}-1} & \text{for $x_1, x_2, \rho_1, \rho_2  \in (0,1)$} \\
  0 & \text{otherwise}
  \end{cases}
$$
\end{itemize}
\end{theorem}

There is the question of choosing the parameter $\rho$.  This vector $\vec \rho$ is not necessarily related to $\vec x$; it can be carefully chosen to obtain different ``shapes'' of negative correlation.  Intuitively, $\rho_e$ controls how much anti-correlation $e$ should have with other edges. We emphasize that there is no single optimal choice, and our scheduling algorithm will use this flexibility in a somewhat subtle way. However, there are a few natural parameterizations, as given in the following result:
\begin{corollary}
\label{cor333}
For a left-node $u$,  let $\lambda_u = \sum_{e \in \Gamma(u)} x_e$.  If we set $\rho_e = 1 - \ee^{-x_e / \lambda_u}$ for all $e \in \Gamma(u)$, then for any edges $e_1, e_2$ with $x_{e_1} x_{e_2} > 0$, we have\footnote{Here and throughout we write $\ee = 2.718...$ We use a different font to distinguish the constant $\ee$ from an edge $e$ in a graph.}
  $$
\Phi(x_{e_1}, x_{e_2}; \rho_{e_1}, \rho_{e_2})  \geq 1 -  \frac{  \ee^{x_{e_1}/\lambda_u} + \ee^{x_{e_2}/\lambda_u}  }{1 + \ee^{1/\lambda_u}}
$$

Moreover, if $\lambda_u \leq 3/4$ and we set $\rho_e = x_e/\lambda_u$ for all $e \in \Gamma(u)$, then 
$$
\Phi(x_{e_1}, x_{e_2}; \rho_{e_1}, \rho_{e_2}) \geq  \frac{\ee^{1/\lambda_u} - 1}{\ee^{1/\lambda_u} + 1} -0.57 (x_{e_1} + x_{e_2}) \max\{0, \lambda_u - 0.45 \} 
$$
\end{corollary}

The first result strictly generalizes the result of \cite{im} (which only considered the setting with $\lambda_u = 1$).  The second result can be stronger quantitatively, and also has a simpler algebraic formula.

We note that there is a closely related setting, as originally considered in \cite{srin1,srin,im}, where we have a bipartite graph $G$ where the edges incident to each left-node are partitioned into ``blocks.'' In this setting, we should have nonpositive-correlation among edges incident to each left-node, and strong negative correlation within each block. Our framework can also handle this setting, by  transforming each block of the graph $G$ to a new separate left-node. Note that the neighborhood of a vertex in the original graph $G$ will form a stable edge-set in the resulting graph $G'$.

  We emphasize that the dependent rounding algorithm, and its analysis, is completely self-contained and does not depend on the machine-scheduling setting. Because of its generality, it may be applicable to other combinatorial optimization problems.

\subsection{Machine-scheduling on unrelated machines}
This problem is denoted $R || \sum_j w_j C_j$ in the common nomenclature for machine scheduling. Here, we have a set of machines $\mathcal M$ and set of jobs $\mathcal J$, where each job $j$ has a weight $w_j$ and has a separate processing time $p^{(i)}_j$ on each machine $i$. Our goal is to assign the jobs to the machines in order, so as to minimize the overall weighted completion time $\sum_j w_j C_j$, where $C_j$ is the sum of $p^{(i)}_{j'}$ over all jobs assigned on machine $i$ up to and including $j$.

On a single machine, there is a simple greedy heuristic for this problem: jobs should be scheduled in non-increasing order of the ratio $\sigma_j = w_j/p_j$, which is known as the \emph{Smith ratio} \cite{smith}. For multiple machines, it is an intriguing and long-studied APX-hard problem \cite{hoogeven}. It has attracted attention, in part, because it leads to  sophisticated convex programming relaxations and rounding algorithms.  Since the 2000's, there were a series of $1.5$-approximation algorithms based on various non-trivial fractional relaxations \cite{schulz1,skutella1,squillante}. 

As shown in \cite{srin1}, going beyond this approximation ratio demands much more involved algorithms: some of the main convex relaxations have integrality gap $1.5$, and furthermore rounding strategies which treat each job independently, as had been used in all previous algorithms, are inherently limited to approximation ratio $1.5$. In a breakthrough, \cite{srin1} achieved a $1.5 - \eps$ approximation factor for some minuscule constant $\eps > 0$. There have been further improvements \cite{li,im,shi2}; most recently before this work, \cite{shi2} gave a $1.45$-approximation (see also \cite{shi3} for a faster implementation of that algorithm). 

These newer algorithms can all be described in the same general framework. First, they solve an appropriate relaxation, giving fractional assignments $x^{(i)}_j$. Second, they group the jobs on each machine, forming ``clusters'' of jobs with similar processing times. Finally, they apply some form of dependent rounding with strong negative correlation properties within each cluster to convert this into an integral assignment.  

There are two main types of relaxations for the first step. The work \cite{srin1} used a semidefinite-programming (SDP) relaxation. The works \cite{li,im,shi2} used a relaxation based on a time-indexed LP. It is not clear which of these relaxations is better, for example, is faster to solve, has a better approximation algorithm, or has a smaller integrality gap. We note that the time-indexed LP may be able to accommodate some other variants of the problem, such as having ``release times'' for each job, which do not seem possible for the SDP relaxation.

There have also been many different types of rounding algorithms used. The work \cite{srin1} used a method based on a random walk in a polytope, which was also used as a black-box subroutine by \cite{li}. This rounding method was subsequently improved in \cite{srin}. The work \cite{im} developed a very different rounding algorithm based on contention resolution of Poisson processes. Finally, \cite{shi2} used a specialized rounding algorithm closely tied to the time-indexed LP structure; it borrows some features from both the random-walk and contention-resolution algorithms.

Our new algorithm will also fit into this framework, where we use the SDP as a starting point and we use our algorithm \textsc{DepRound} for the rounding. We show the following:
\begin{theorem}
\label{thm3}
There is a randomized approximation algorithm for Scheduling on Unrelated Machines to Minimize Weighted Completion Time with approximation ratio $\rhoValue$. 

In particular, the SDP relaxation  has integrality gap at most $\rhoValue$.
\end{theorem}

This is the first improvement to the integrality gap for the SDP relaxation since the original work \cite{srin1}; all later improvements have been based on the time-indexed LP.  It is interesting now to determine which fractional relaxation has a better gap.

To explain our improvement, we note that it is relatively straightforward  to handle a scenario with many jobs of small mass and similar processing times.  The difficulty lies in grouping disparate items together. This typically results in ``ragged'' clusters and, what is worse, the ``leftover clusters'' for each processing-time class.  A large part of our improvement is technical, coming from tracking the contributions of leftover clusters more carefully. There are two algorithmic ideas to highlight. 

First, the \textsc{DepRound} algorithm is much more flexible for rounding: its negative correlation guarantees scale with the size of each cluster or even the individual items, as opposed to the \emph{maximum} cluster size. The parameter $\rho$ available in \textsc{DepRound} allows us to put less ``anti-correlation strength'' on the final job within a cluster. This makes the clusters act in a significantly more uniform way irrespective of their total mass. In particular, we no longer need to  deal separately with jobs which have ``large'' mass on a given machine.

Second, we use a random shift before quantizing the items by processing time, which ensures that items are more evenly distributed within each class. This technique was also used in \cite{im,shi2}, but had not been analyzed in the context of the SDP relaxation.

\paragraph{Note:} After the conference version of this paper, \cite{shi-new} proposed a new algorithm for this problem with approximation ratio $1.36$. It uses different techniques, based on iterated rounding. Its  approximation ratio (and integrality gap) are shown in terms of the Configuration LP, which is stronger than either the time-indexed LP or SDP relaxations.

\section{Negatively-correlated Exponential random variables}
\label{round-sec1}
Before we consider bipartite rounding, we derive a more basic result in probability theory: how to generate Exponential random variables with strong negative correlations.  The analysis relies heavily on properties of \emph{negatively-associated} (NA) random variables \cite{joag-dev}. Formally, we say that random variables $X_1, \dots, X_k$ are NA if for any disjoint subsets $A, B \subseteq \{1, ..., k \}$ and any increasing functions $f,g$, there holds
\begin{equation}
\label{ata}
\bE[ f( X_i: i \in A) g (X_j: j \in B) ] \leq \bE[ f( X_i: i \in A) ] \bE[g( X_j: j \in B) ]
\end{equation}

We quote a few useful facts about such variables from \cite{joag-dev, wajc}.
\begin{theorem}
\label{nathm}
\begin{enumerate}
\item If $X_1, \dots, X_k$ are NA, then $\bE[ X_1 \cdots X_k ] \leq \bE[X_1] \cdots \bE[X_k]$.
\item  If $X_1, \dots, X_k$ are zero-one random variables with $X_1 + \dots + X_k \leq 1$, then $X_1, \dots, X_k$ are NA.
\item If $\mathcal X, \mathcal X'$ are collections of NA random variables, and the joint distribution of $\mathcal X$ is independent from that $\mathcal X'$, then $\mathcal X \cup \mathcal X'$ are NA.
\item If $X_1, \dots, X_k$ are NA random variables and $f_1, \dots, f_{\ell}$ are functions defined on disjoint subsets of $\{1, \dots, k \}$, such that all $f_1, \dots, f_{\ell}$ are monotonically non-increasing or all $f_1, \dots, f_{\ell}$ are monotonically non-decreasing, then random variables $f_i(\vec X): i  = 1, \dots, \ell$ are NA. In particular, $\bE[ f_1(\vec X) \cdots f_{\ell}( \vec X)] \leq \bE[ f_1(\vec X) ] \cdots \bE[ f_{\ell} (\vec X) ]$.
\end{enumerate}
\end{theorem}

Our algorithm also makes use of the \emph{Multivariate Geometric distribution} \cite{mvg}: given a probability vector $\vec p \in [0,1]^n$, this distribution on vector $(X_1, \dots, X_n)$ is defined by running an infinite sequence of experiments, where each $j^{\text{th}}$ experiment takes on value $i$ with probability $p_i$, and setting $X_i$ to be the number of trials before first seeing $i$. We note a few simple properties of this distribution.
\begin{observation}
\label{multivariate-obs}
Let $(X_1, \dots, X_n)$ be a  Multivariate Geometric random variable with rate $\vec p$. Then the variables $X_1, \dots, X_n$ are NA and each marginal distribution $X_i$ is Geometric with rate $p_i$.
\end{observation}
\begin{proof}
Observe that $X_1, \dots, X_n$ can be viewed in terms of the following infinite sequence of random variables: for each $j \in \mathbb Z_{\geq 0}$, draw variable $F_{j} \in \{1, \dots, n \}$ with $\Pr(F_j = i) = p_i$ for all $i,j$. Then set $X_i = \min \{ j: F_j = i \}$ for each $i$. Let $B_{i,j}$ be the indicator random variable for the event $F_j = i$. For each $j$, the values $B_{i,j}$ are zero-one random variables which sum to one, so they are NA. Since there is no interaction between the different indices $j$, all random variables $B_{i,j}$ are NA. Furthermore, each $X_i$ is a monotone-down function of the values $B_{i,j}$. Hence, all variables $X_i$ are NA.

Furthermore, each $X_i$ is the waiting time until we see the first variable $B_{i,j}$; since these are Bernoulli with rate $p_i$, the distribution of $X_i$ is Geometric with rate $p_i$
\end{proof}

We describe our algorithm to generate correlated unit Exponential random variables $Z_1, \dots, Z_n$. It takes as input a vector $\vec \rho \in [0,1]^n$ with $\sum_i \rho_i \leq 1$, which will determine the covariances.

\begin{algorithm}[H]
\caption{\sloppy {$\textsc{CorrelatedExponential}(\vec \rho)$}}
Set $\rho_0 = 1 - (\rho_1 + \dots + \rho_n)$ \\
Set $(X_0, \dots, X_n) \leftarrow \textsc{MultivariateGeometric}(\vec \rho)$. \\
\For{$i = 1, \dots, n$}{
\If{$\rho_i \in (0,1)$} {
Set $\alpha_i = - \log (1 - \rho_i)$ \\
Draw random variable $S_i \in [0,1]$ with probability density function $\alpha_i \ee^{-\alpha_i s} / \rho_i$ \\
Set $Z_i = \alpha_i ( X_i + S_i )$ \\
} \Else {
Draw $Z_i$ as an independent Exponential random variable with rate $1$
}
}
\Return vector $Z$
\end{algorithm}

The density function at Line 5 is valid in that $\int_{s=0}^1 \alpha_i \ee^{-\alpha_i s} / \rho_i  \ \mathrm{d}s = \frac{1}{\rho_i} (1 - \ee^{-\alpha_i}) = 1$.  The algorithm can be implemented to run in randomized polynomial time via standard sampling procedures.  We now analyze its probabilistic properties.
\begin{proposition}
Each random variable $Z_i$ has an Exponential distribution with rate $1$, i.e. it has probability density function $e^{-z}: z \in [0, \infty)$.
\end{proposition}
\begin{proof}
It is clear if $\rho_i \in \{0,1 \}$, so suppose $\rho_i \in (0,1)$ and consider random variable $L = X_i + S_i$. We claim that $L$ has an Exponential distribution with rate $\alpha_i$. For, consider some $z \in \mathbb R_{\geq 0}$, where $z = j + s$ and $j \in \mathbb Z_{\geq 0}, s \in [0,1)$. Since $X_i$ is Geometric with rate $\rho_i$, the probability density function for $L$ at $z$ is given by
$$
\rho_i (1-\rho_i)^j \cdot \frac{\alpha_i \ee^{-\alpha_i s}}{\rho_i}  = (\ee^{-\alpha_i})^j \cdot \alpha_i \ee^{-\alpha_i s} = \ee^{-\alpha_i (j+s)} \cdot \alpha_i = \alpha_i \cdot \ee^{-\alpha_i z}
$$
which is precisely the pdf of a rate-$\alpha_i$ Exponential random variable. Since rescaling Exponential random variables changes their rate, the variable $Z_i$ is Exponential with rate $1$.
\end{proof}

\begin{proposition}
The random variables $Z_1, \dots, Z_n$ are NA.
\end{proposition}
\begin{proof}
The random variables $X_i$ are NA by Observation~\ref{multivariate-obs}. The $S_i$ variables are clearly NA since they are independent. Each variable $Z_i$ for $\rho_i \in (0,1)$ is an increasing function of the random variables $S_i, X_i$. Hence, these $Z$ variables are also NA. The $Z_i$ variables with $\rho_i \in \{0,1 \}$ are NA since they independent of all other variables.
\end{proof}

The crucial property is that the variables $Z_i$ have significant negative correlation (depending on shape parameter $\vec \rho$). Specifically, we have the following:
\begin{lemma}
\label{explemmabb1}
For any indices $i_1, i_2$ with $\rho_{i_1}, \rho_{i_2} \in (0,1)$, and any values $q_1, q_2 \in (-\infty,1)$, we have
$$
\bE[ \ee^{q_1 Z_{i_1} + q_2 Z_{i_2}} ] = \frac{1}{(1-q_1) (1-q_2)} \cdot \Bigl( 1 - \frac{ ( (1 - \rho_{i_1})^{q_1} - 1) ( (1 - \rho_{i_2})^{q_2} - 1) }{
 (1 -\rho_{i_1})^{q_1}
  (1-\rho_{i_2})^{q_2}+\rho_{i_1}+\rho_{i_2}-1} \Bigr)
$$
\end{lemma}
\begin{proof}
Suppose without loss of generality that $i_1 = 1, i_2 = 2$. For $i = 1, 2$ define parameter $\theta_i = \alpha_i q_i  = -q_i \log(1 - \rho_i)$ and random variable $L_i = X_i + S_i$. We calculate:
$$
\bE[ \ee^{q_1 Z_1 + q_2 Z_2} ] = \bE[ \ee^{\theta_1 L_1 + \theta_2 L_2}]  = \bE[ \ee^{\theta_1 S_1}  ] \bE[ \ee^{\theta_2 S_2} ] \bE[ \ee^{\theta_1 X_1 + \theta_2 X_2}]
$$ 
where the last equality holds since the variables $S_1, S_2$ are independent of $X$. 

We will calculate these term by term. For $i = 1, 2$ we have:
\begin{align*}
\bE[ \ee^{\theta_i S_i} ] &= \frac{\alpha_i}{\rho_i} \int_{s=0}^1 \ee^{-\alpha_i s + \theta_i s} \ \mathrm{d}s = \frac{ \alpha_i ( 1 - \ee^{\theta_i - \alpha_i}) }{\rho_i (\alpha_i - \theta_i) }
\end{align*}

We calculate the joint distribution of $X_1, X_2$ as:
$$
\Pr(X_1 = x_1, X_2 = x_2) = \begin{cases}
 (1 - \rho_1 - \rho_2)^{x_1} \rho_1 (1 - \rho_2)^{x_2 - x_1 - 1} \rho_2 & \text{for $0 \leq x_1 < x_2$} \\
 (1 - \rho_1 - \rho_2)^{x_2} \rho_2 (1 - \rho_1)^{x_1 - x_2 - 1} \rho_1 & \text{for $0 \leq x_2 < x_1$} \\
 0 & \text{otherwise}
 \end{cases}
 $$

So we can sum over pairs $x_1, x_2$ to get:
\begin{align*}
\bE[  \ee^{\theta_1 X_1 + \theta_2 X_2}] &= \sum_{0 \leq x_1 < x_2} (1 - \rho_1 - \rho_2)^{x_1} \rho_1 (1 - \rho_2)^{x_2 - x_1 - 1} \rho_2 \ee^{\theta_1 X_1 + \theta_2 X_2} \\
&\qquad \qquad +  \sum_{0 \leq x_2 < x_1}(1 - \rho_1 - \rho_2)^{x_2} \rho_2 (1 - \rho_1)^{x_1 - x_2 - 1} \rho_1 \ee^{\theta_1 X_1 + \theta_2 X_2} \\
&= \frac{\rho_1 \rho_2}{1 - \ee^{\theta_1 + \theta_2} (1 - \rho_1 - \rho_2)} \cdot \Bigl( \frac{1}{\ee^{-\theta_1} + \rho_1 - 1} + \frac{1}{\ee^{-\theta_2} + \rho_2 - 1} \Bigr) 
\end{align*}

The result then follows by multiplying the formulas for $\bE[  \ee^{\theta_1 X_1 + \theta_2 X_2}], \bE[ \ee^{\theta_1 S_1} ], \bE[ \ee^{\theta_2 S_2} ]$ and substituting for $\alpha_i = -\log(1 - \rho_i), \theta_i = -q_i \log(1 - \rho_i)$.
\end{proof}

\section{Bipartite dependent rounding algorithm}
\label{round-sec}
We consider here a variant of a rounding algorithm from \cite{im}, which was in turn inspired by an earlier ``fair sharing'' algorithm of \cite{feige}. For motivation, consider the following natural rounding procedure: each edge $e$ draws an Exponential random $Z_e$ variable with rate $x_e$. Then, for each right-node $v$, we set $X_f = 1$ for the edge $f = \argmin_{e \in \Gamma(v)} X_e$. It is easy to see that this is equivalent to independent rounding, and in particular we have $\Pr(X_e = 1) = x_e$. 

Instead of generating the variables $Z_e$ independently, we will use our algorithm \textsc{CorrelatedExponential}. This creates the desired negative correlation in the bipartite rounding.

In describing our algorithm, we will assume that our original fraction solution $\vec x$ satisfies  $$
\forall e \in E \ \  x_e \in (0,1), \qquad \text{and} \qquad \forall v \in V \sum_{e \in \Gamma(v)} x_e = 1
$$

These conditions can be assumed without loss of generality by removing edges with $x_e \in \{0,1 \}$, and by adding dummy edges for each right-node.

\begin{algorithm}[H]
\caption{\sloppy {$\textsc{DepRound}(\vec \rho, \vec x)$}}
\For{each left-node $u$} {
Call $\vec Z^{(u)} \leftarrow \textsc{CorrelatedExponential}(\vec \rho^{(u)} )$ for vector $(\rho^{u}(v) = \rho_{(u,v)}: v \in \Gamma(u))$
}
Combine all vectors $\vec Z^{(u)}$ into a single vector $Z \in \mathbb R_{\geq 0}^{E}$ with $Z_{(u,v)} = Z^{(u)}_v$ for all $u,v$. \\
\For{each right-node $v$} {
Choose neighbor $e = \argmin_{f \in \Gamma(v)} Z_f/x_f$. \\
Set $X_{e} = 1$ and $X_{e'} = 0$ for all other neighbors $e' \in \Gamma(v) \setminus \{ e \}$ \\
}
\Return vector $\vec X$
\end{algorithm}

 From properties of \textsc{CorrelatedExponential}, the following properties immediately hold:
\begin{proposition}
\label{lemagg1}
Suppose that $\vec \rho \in [0,1]^{E}$ satisfies $\sum_{e \in \Gamma(u)} \rho_e = 1$ for all left-nodes $u$. Then \textsc{DepRound} can be implemented in polynomial time with the following properties:
\begin{itemize}
\item Each variable $Z_e$ has the Exponential distribution with rate $1$. 
\item The random variables $Z_e: e \in E$ are NA.
\item For a right-node $v$, all random variables $Z_{e}: e \in \Gamma(v)$ are independent.
\end{itemize}
\end{proposition}

For an edge-set $L$, let us write $Z_L$ for the vector of random variables $(Z_e: e \in L)$.  At this point, we can immediately show property (A2):
\begin{proposition}
For any edge $e$ there holds $\Pr(X_e = 1) = x_e$.
\end{proposition}
\begin{proof}
The random variables $Z_f: f \in \Gamma(v)$ are independent unit-rate Exponentials. So random variables $Z_f/x_f: f \in \Gamma(v)$ are independent Exponentials with rates $x_f$ respectively, and $\sum_{f \in \Gamma(v)} x_f = 1$.   It is a well-known standard fact that, for independent Exponential random variables $Y_1, \dots, Y_{\ell}$ with rates $\lambda_1, \dots, \lambda_{\ell}$, there holds $\Pr( Y_i = \min\{ Y_1, \dots, Y_{\ell} \}) = \frac{\lambda_i}{\lambda_1 + \dots + \lambda_{\ell}}$.
\end{proof} 

We next turn to show properties (A3) and (A4). For this, we have the following key lemma:
\begin{lemma}
\label{gatt4}
Let $L \subseteq E$ be a stable edge-set. If we reveal the random variables $Z_L$, then we have
$$
\bE \Bigl[ \prod_{e \in L}  X_e  \mid Z_L \Bigr] \leq \prod_{e \in L} \ee^{(1 - 1/x_e) \cdot  Z_e}
$$
\end{lemma}
\begin{proof}
We assume that all right-nodes of edges in $L$ are distinct, as otherwise $\prod_{e \in L} X_e = 0$ with probability one.

Define $L'$ to be the set of edges outside $L$ which share a right-node with an edge in $L$, that, is, the set of edges of the form $f = (u,v) \notin L$ where $(u', v) \in L$.  Let $W, W'$ denote the set of left-nodes of edges of $L, L'$ respectively. We claim that $W, W'$ are disjoint. For, suppose $u \in W \cap W'$. So there edges $(u, v) \in L, (u,v') \in L'$; by definition of $L'$, there must be a corresponding edge $(u', v') \in L$. Since $G$ is a simple graph, necessarily $u \neq u'$. Also $v \neq v'$ since edges in $L$ have distinct right-nodes. This contradicts that $L$ is a stable set.

Suppose that we condition on all random variables corresponding to the nodes in $W$, in particular, we reveal all values $Z_L$. The random variables corresponding to nodes in $W'$ have their original unconditioned probability distributions. We now have
$$
\Pr( \bigwedge_{e \in L} X_e = 1 ) = \Pr \Bigl( \bigwedge_{(u,v) \in L} Z_{(u,v)}/x_{(u,v)} = \min_{f \in \Gamma(v)} Z_f/x_f \Bigr)
$$
where, here and in the remainder of the proof, we omit the conditioning on $W$ for brevity. 

Each event $Z_{(u,v)}/x_{(u,v)} = \min_{f \in \Gamma(v)} Z_f/x_f$ is an increasing function of random variables $Z_{\Gamma(v) \setminus L}$. These sets $\Gamma(v) \setminus L$ are disjoint since the right-endpoints of edges in $L$ are all distinct.  Since random variables $Z_{L'}$ are NA, Theorem~\ref{nathm} yields
$$
\Pr \Bigl( \bigwedge_{(u,v) \in L} Z_{(u,v)}/x_{u,v} = \min_{f \in \Gamma(v)} Z_f/x_f \Bigr) \leq \prod_{(u,v) \in L}  \Pr \bigl( Z_{(u,v)}/x_{(u,v)} = \min_{f \in \Gamma(v)} Z_f/x_f  \bigr)
$$

For any edge $e = (u,v) \in L$, the variables $Z_f: f \in \Gamma(v) \setminus \{ e \}$ are independent unit-rate Exponentials. By standard facts about Exponential random variables, this implies that $Z' := \min_{f \in \Gamma(v) \setminus \{e \}} Z_f/x_f$ is  an Exponential random variable with rate $\sum_{f \in \Gamma(v) \setminus \{e\}} x_f = 1 - x_e$. The probability that $Z' > Z_e/x_e$ is precisely $\ee^{-(1-x_e) \cdot Z_e / x_e} = \ee^{ (1 - 1/x_e) Z_e } $.
\end{proof}

\begin{proposition}
Property (A3) holds.
\end{proposition}
\begin{proof}
Consider a  stable edge set $L \subseteq E$. By iterated expectations with respect to random variable $Z_L$ and Lemma~\ref{gatt4}, we have
$$
\bE \Bigl[ \prod_{e \in L} X_e \Bigr] = \bE_{Z_L} \Bigl[  \bE \bigl[ \prod_{e \in L} X_e \mid Z_L \bigr] \Bigr] \leq \bE \Bigl[ \prod_{e \in L} \ee^{(1 - 1/x_e) Z_e} \Bigr]
$$

Each term $\ee^{(1 - 1/x_e) Z_e}$ in this product is a decreasing function of random variable $Z_e$. Since the variables $Z_L$ are NA, Theorem~\ref{nathm} gives:
$$
\bE \Bigl [  \prod_{e \in L} \ee^{(1 - 1/x_e) Z_e} \Bigr]  \leq \prod_{e \in L}\bE \bigl[ \ee^{(1 - 1/x_e) Z_e} \bigr]
$$

Here, for an edge $e$, we have $\bE[ \ee^{(1-1/x_e) Z_e} ] = \int_{z=0}^{\infty} {\ee}^{-z} \cdot  \ee^{(1-1/x_e) z} \ \mathrm{d}z = x_e.$
\end{proof}

\begin{theorem}
\label{gatt5}
For any edges $ e_1 = ( u, v_1),  e_2 = ( u, v_2)$ with the same left-node $u$, there holds $\bE[ X_{e_1}  X_{e_2}] \leq   x_{e_1}   x_{e_2} \cdot ( 1 -  \Phi(x_{e_1}, x_{e_2}; \rho_{e_1}, \rho_{e_2}) ) $.
\end{theorem}
\begin{proof}
Let us write $ X_i = X_{ e_i},  x_i = x_{ e_i},  \rho_i = \rho_{e_i},  Z_i = Z_{e_i}$ for $i = 1,2$.  By Lemma~\ref{gatt4} applied to stable-set  $L = \{  e_1,  e_2 \}$ we have:
\begin{align*}
\bE[   X_1  X_2 \mid  Z_1,  Z_2] &\leq \ee^{(1- 1/x_1)  Z_1 + (1- 1/x_2)  Z_2} 
\end{align*}

Observe that values $Z_1, Z_2$ are simultaneously generated by $\textsc{CorrelatedExponential}(\vec \rho^{(u)})$.  By applying Lemma~\ref{explemmabb1} with $q_1 = 1 - 1/x_1, q_2 = 1 - 1/x_2$, we have
\[
\bE[ \ee^{(1- 1/x_1)  Z_1 + (1- 1/x_2)  Z_2}  ] = x_1 x_2 \cdot
\Bigl( 1 -  \frac{ (  (1-\rho_1)^{1 - 1/x_1} - 1) ( (1-\rho_2)^{1-1/x_2} - 1)}{
 (1 -\rho_1)^{1-1/x_1}
  (1-\rho_2)^{1-1/x_2}+\rho_1+\rho_2-1} \Bigr). \qedhere
\]
\end{proof}

As we have mentioned, there is no single optimal choice for the vector $\vec \rho$. One attractive option is to simply set $\vec \rho$ proportional to $\vec p$, which leads to the following simple version of property (A4):
\begin{lemma}
\label{gatt7}
For $t \geq 4/3$ and $x_1, x_2 \in (0,1)$, there holds
$$
\Phi( x_1, x_2; t x_1, t x_2 ) \geq 1 - \frac{\ee^{t} - 1}{\ee^{t} + 1} + 0.57 (x_1 + x_2) \max\{0, 1/t - 0.45 \} 
$$
\end{lemma}

 The proof of Lemma~\ref{gatt7} involves significant numerical analysis; we defer it to Appendix~\ref{gatt7app}.    In the scheduling algorithm, \emph{most} edges will set $\rho_e \propto x_e$, and use the bound in Lemma~\ref{gatt7}. However, a few edges will use a different value of $\rho_e$. This is a good illustration of the flexibility of the bipartite dependent-rounding scheme. 
 
 Another nice choice, which can work for more general graphs, is the following:
\begin{proposition}
\label{gatt6}
For $t > 0$ and $x_1, x_2 \in (0,1)$, there holds
$$
\Phi(x_1, x_2; 1 - \ee^{-x_1 t}, 1 - \ee^{-x_2 t}) =  1 - \frac{ ( \ee^{t}-1) ( \ee^{x_1 t} + \ee^{x_2 t})}{ \ee^{2 t} -  \ee^{x_1 t} - \ee^{x_2 t} + \ee^{(x_1 + x_2) t}}\geq 1 - \frac{ \ee^{x_1 t} + \ee^{x_2 t}}{1 + \ee^{t}}
$$
\end{proposition}
\begin{proof}
Let $\rho_i = 1 - \ee^{x_i t}$ for $i = 1, 2$, and let $\eta = \ee^t$. With some algebraic simplifications we get
$$
\Phi(x_1, x_2; \rho_1, \rho_2) =  1 - \frac{ ( \ee^{t}-1) ( \ee^{x_1 t} + \ee^{x_2 t})}{ \ee^{2 t} -  \ee^{x_1 t} - \ee^{x_2 t} + \ee^{(x_1 + x_2) t}} = 1 - \frac{ (\eta-1) ( \eta^{x_1} + \eta^{x_2})}{ \eta^2 - \eta^{x_1} - \eta^{x_2} + \eta^{x_1 + x_2}} 
$$
Furthermore, we can observe that
\[
\frac{ ( \eta-1) ( \eta^{ x_1} +  \eta^{ x_2})}{ \eta^2 -  \eta^{ x_1} -  \eta^{ x_2} +  \eta^{ x_1 +  x_2}}  =  \frac{ (\eta-1) (\eta^{x_1} + \eta^{ x_2})}{ (\eta^{x_1} - 1)(\eta^{x_2} - 1) + (\eta^2-1) }\leq \frac{ (\eta-1) (\eta^{x_1} + \eta^{ x_2})}{(\eta^2-1) } =  \frac{ \eta^{x_1} + \eta^{ x_2}}{\eta+1}. \qedhere
 \]
\end{proof}

\section{Scheduling to minimize completion time}
\label{machine-sched-sec}

Our general algorithm can be summarized as follows:
\begin{itemize}
\item Solve the SDP relaxation, obtaining estimates $x^{(i)}_{j,j'}$ for each machine $i$ and pair of jobs $j, j'$.  Roughly speaking, $x^{(i)}_{j,j'}$ represents the fractional extent to which jobs $j, j'$ are simultaneously scheduled on machine $i$. The case with $j = j'$ plays an especially important role, in which case we write simply $x^{(i)}_j$.
\item Based on the fractional  solution $\vec x$, partition the jobs on each machine $i$ into clusters $\mathcal C^{(i)}_{1}, \dots, \mathcal C^{(i)}_{t}$.
\item Run \textsc{DepRound} to obtain rounded variables $X^{(i)}_{j} \in \{0, 1 \}$, with $\sum_{i} X^{(i)}_{j} = 1$ for all jobs $j$. Job $j$ is assigned to the machine $i$ with $X^{(i)}_j = 1$.
\item Schedule jobs assigned to each machine $i$ in non-increasing order of Smith ratio $\sigma^{(i)}_j = w_j/p^{(i)}_j$.
\end{itemize}

We will not modify the first or last steps in any way; they are discussed in more detail in Section~\ref{sdp-subsec} next. The  difference is how we implement the second step and third steps (partitioning and rounding the jobs). Specifically,  given a fractional solution $x$, we form our clusters as follows:

\begin{algorithm}[H]
\caption{Clustering the jobs}
Define parameters  $\pi = \piValue,  \theta = \thetaValue ,  \lmax = \lmaxValue$.

Draw a random variable $\roff$ uniformly at random from $[0, 1]$.

\For{each machine $i$} {

Partition the jobs into processing time classes $\mathcal P^{(i)}_k = \{ j:  \roff + \frac{\log p^{(i)}_j}{\log \pi} \in [k, k+1)  \}$.

\For{each class $\mathcal P^{(i)}_k$} {

Initialize cluster index $\ell = 1$ and set $\mathcal C^{(i)}_{k,1} = \emptyset$.

Sort the jobs in $\mathcal P^{(i)}_k$ with $x^{(i)}_j > 0$ in non-increasing order of Smith ratio as $j_1, j_2, \dots, j_{s}$.

\For{$t = 1, \dots, s$ and each job $j = j_t$}  {

Set $\tilde \rho^{(i)}_j = \min\{ x^{(i)}_j,    \lmax - \sum_{j' \in \mathcal C^{(i)}_{k,\ell}} x^{(i)}_{j'} \}$

Update $\mathcal C^{(i)}_{k,\ell} \leftarrow \mathcal C^{(i)}_{k,\ell} \cup \{ j_t \}$. 

\If{ $\sum_{j \in \mathcal C^{(i)}_{k,\ell}} x^{(i)}_j \geq \theta$} {

Update $\ell \leftarrow \ell + 1$ and initialize the new cluster $\mathcal C^{(i)}_{k,\ell} = \emptyset$

}
}

For each cluster $k,\ell$ and every job $j \in \mathcal C^{(i)}_{k,\ell}$ set $\rho^{(i)}_j = \frac{ \tilde \rho^{(i)}_j } {\sum_{j' \in \mathcal C^{(i)}_{k,\ell}} \tilde \rho^{(i)}_{j'} }$.

}

Run  $\textsc{DepRound}( \rho,  x)$ to convert the fractional solution $x$ into an integral solution $X$.

}

\end{algorithm}

So each machine has a single ``open'' cluster for each processing-time quantization class at a time. If, after adding the job to the open cluster, the cluster size becomes at least $\theta$, then we close it and open a new one. The correlation parameter $\rho^{(i)}_j$ is usually chosen to be proportional to $x^{(i)}_j$; the one exception is that the final job that closes out a cluster may need to choose a smaller parameter $\rho^{(i)}_j$.  We say a job $j$ is \emph{truncated} on a machine $i$ if $\tilde \rho_j^{(i)} < x^{(i)}_j$; note that, in this case, every other job $j'$ in the cluster has $\rho_{j'}^{(i)} = x^{(i)}_j / \tau$.

 Note that, in our rounding algorithm, we only use the diagonal terms $x^{(i)}_j$ of the SDP relaxation; the ``cross-terms'' $x^{(i)}_{j,j'}$ appear only for the analysis. We also remark that, unlike prior algorithms, there is no special handling for ``large'' jobs (jobs with large mass $x^{(i)}_j$).

\begin{proposition}
\label{phi-mr-prop}
Algorithm \textsc{DepRound} can be used to round the fractional solution $x$ to an integral solution $X$, where for any machine $i$ and distinct jobs $j, j'$, it satisfies
$$
\bE[X^{(i)}_j] = x^{(i)}_j, \qquad \bE[ X^{(i)}_j X^{(i)}_{j'} ] \leq x^{(i)}_j x^{(i)}_{j'}
$$
and for any machine $i$ and any cluster $k, \ell$ and distinct jobs $j, j' \in \mathcal C^{(i)}_{k,\ell}$, it has
$$
\bE[ X^{(i)}_{j} X^{(i)}_{j'} ] \leq (1 - \phi^{(i)}_{j, j'}) \cdot x^{(i)}_j x^{(i)}_{j'}  \quad \text{ for $
\phi^{(i)}_{j, j'} := \Phi( x^{(i)}_{j}, x^{(i)}_{j'}; \rho^{(i)}_j, \rho^{(i)}_{j'})$}
$$
\end{proposition}
\begin{proof}
Apply Theorem~\ref{thm1} to the graph with left-nodes $\mathcal C^{(i)}_{k,\ell}$ and right-nodes $\mathcal J$; for each machine $i$ and job $j$ where $j \in \mathcal C^{(i)}_{k,\ell}$, it has an edge $e = ( \mathcal C^{(i)}_{k,\ell}, j)$ with $x_e = x^{(i)}_j, \rho_e = \rho^{(i)}_j, X^{(i)}_j = X_e$.  Note that $\sum_{e \in \Gamma(j)} x_e = \sum_{e \in \mathcal C^{(i)}_{k,\ell}} \rho_e = 1$.

By Property (A1), every job is assigned to exactly one machine. By Property (A2), we have $\bE[X^{(i)}_j] = x^{(i)}_{j}$. For any pair of jobs $j, j'$ on a machine $i$, note that the corresponding edge-set  $\{   ( \mathcal C^{(i)}_{k,\ell}, j),    ( \mathcal C^{(i)}_{k',\ell'}, j') \}$ is stable. So  Property (A3) gives $\bE[ X^{(i)}_j X^{(i)}_{j'} ] \leq x^{(i)}_j x^{(i)}_{j'}$. Also, by Property (A4), every pair of jobs $j, j'$ in a  cluster $\mathcal C^{(i)}_{k,\ell}$  has  $\bE[ X^{(i)}_j X^{(i)}_{j'} ] \leq  (1-\phi^{(i)}_{j,j'}) x^{(i)}_{j} x^{(i)}_{j'}$.
 \end{proof}
 
For each class $\mathcal P^{(i)}_k$ and each job $j \in \mathcal P^{(i)}_k$, we define $$
P_k = \pi^{k - \roff}, \qquad H^{(i)}_j = p^{(i)}_j / P_k.
$$

Note that $H^{(i)}_j \in [1,\pi]$ for all $j$.  The key property we exploit is that each random value $\log_{\pi} H^{(i)}_j$ is uniformly distributed in $[0, 1]$. This leads to the following formula:
\begin{observation}
\label{hest-prop}
For any job $j$ and any function $\Psi: [1,\pi] \rightarrow \mathbb R$, we have
$$
\bE[ \Psi(H_j)  ] = \frac{1}{\log \pi} \int_{h = 1}^{\pi} \frac{\Psi(h)}{h}  \ \mathrm{d}h
$$
\end{observation}

\subsection{The SDP relaxation}
\label{sdp-subsec}
Here, we provide a brief summary of the SDP relaxation and its properties. See \cite{srin1} for more details.

For a \emph{single} machine, there is a simple heuristic to minimize weighted completion time: namely, jobs should be ordered in non-increasing order of their Smith ratio $\sigma^{(i)}_j$. To simplify notation, let us suppose that all the ratios $\sigma^{(i)}_j$ are distinct. (This can be achieved without loss of generality by adding infinitesimal noise to each $p^{(i)}_j$.)  Thus, the overall weighted completion time would be
$$
\sum_{\substack{j, j' \text{assigned to machine $i$} \\  \sigma^{(i)}_{j'} \leq \sigma^{(i)}_j}} w _j  p^{(i)}_{j'} 
$$

For each machine $i$, let us further define a symmetric $(|\mathcal J|+1) \times (|\mathcal J|+1)$ matrix $\xmat^{(i)}$ as follows: we set $\xmat^{(i)}_{0,0} = 1$, we set $\xmat^{(i)}_{0,j} = \xmat^{(i)}_{j,0} = x^{(i)}_{j}$, and we set $\xmat^{(i)}_{j, j'} = x^{(i)}_{j,j'}$ for all pairs of jobs $j, j'$. This motivates the following semidefinite-programming (SDP) relaxation:
\begin{align*}
\text{maximize} &\qquad \qquad  \sum_{i \in \mathcal M} \sum_{j \in \mathcal J} w_j  \sum_{\substack{j' \in \mathcal J: \sigma^{(i)}_{j'} \leq \sigma^{(i)}_j}} p^{(i)}_{j'} x^{(i)}_{j, j'} \\
\text{subject to} & \qquad \qquad x^{(i)}_{j, j'} \in [0,1] \qquad \text{for all $j, j'$} \\
& \qquad \qquad \sum_{i \in \mathcal M} x^{(i)}_j = 1 \qquad \text{for all $j$} \\
& \qquad  \text{Each matrix $\xmat^{(i)}$ is symmetric and positive-semidefinite}
\end{align*}

As described in \cite{srin1}, this relaxation can be solved in polynomial time. Furthermore, given an integral solution $X^{(i)}_{j} \in \{0,1 \}^{\mathcal J}$, there is a corresponding SDP solution defined by $x^{(i)}_{j,j'} = X^{(i)}_j X^{(i)}_{j'}$. Our goal is to convert the fractional assignments $x^{(i)}_j$ into integral assignments $X^{(i)}_j$. We quote the following key results of \cite{srin1} concerning the semidefinite program.
\begin{theorem}[\cite{srin1}]
\label{thm:partitionrel}
For any machine $i^*$, suppose the jobs are sorted $1, \dots, n$ in non-increasing order of Smith ratio $\sigma^{(i^*)}_j$ along with a dummy job $n+1$ with $w_{n+1} = \sigma^{(i^*)}_{n+1} = 0$.  Then the completion time on machine $i^*$ is
$$
\sum_{j^*=1}^n ( \sigma^{(i^*)}_{j^*} - \sigma^{(i^*)}_{j^*+1}) Z^{(i^*,j^*)}
$$
while the contribution to the SDP objective function corresponding to machine $i^*$ is given by
$$
\sum_{j^*=1}^n (\sigma^{(i^*)}_{j^*} - \sigma^{(i^*)}_{j^*+1}) \lb^{(i^*,j^*)}
$$
where for each job $j^*$ we define 
\begin{align*}
Z^{(i^*,j^*)} &= \frac{1}{2} \Bigl( \sum_{j=1}^{j^*} X_j^{(i^*)} (p_{j}^{(i^*)})^2  + \sum_{j=1}^{j^*} \sum_{j'=1}^{j^*}  X_j^{(i^*)} X_{j'} ^{(i^*)} p_j^{(i^*)} p_{j'}^{(i^*)}   \Bigr)  \\
\lb^{(i^*,j^*)} &= \frac{1}{2} \Bigl( \sum_{j=1}^{j^*} x_j^{(i^*)} (p_{j}^{(i^*)})^2 + \sum_{j=1}^{j^*} \sum_{j'=1}^{j^*} x^{(i^*)}_{j,j'} p^{(i^*)}_j p^{(i^*)}_{j'}  \Bigr) 
\end{align*}
\end{theorem}

\begin{corollary}
\label{cor:partitionrel3}
If $\bE[ Z^{(i^*,j^*)} ] \leq \eta \cdot \lb^{(i^*,j^*)}$ for all machines $i^*$ and jobs $j^*$, then the resulting schedule is an $\eta$-approximation in expectation, and the SDP has integrality gap at most $\eta$.
\end{corollary}

\subsection{Focusing on a single machine and job.}
The main consequence of Corollary~\ref{cor:partitionrel3} is that we can focus on a single machine $i^*$ and single job $j^*$  and ignore all the job weights. \emph{For the remainder of the analysis, we suppose $i^*, j^*$ are fixed, and we omit all superscripts $(i^*, j^*)$,} for example we write $X_{j}$ instead of $X^{(i^*)}_{j}$. 

We define $\mathcal J^*$ to be the set of jobs $j$ with $x_j > 0$ and $\sigma_{j} \geq \sigma_{j^*}$. As a point of notation, any sum of the form $\sum_j$ should be taken to range over $j \in \mathcal J^*$ unless stated otherwise. Likewise, in a sum of the form $\sum_{j, j'}$, we view $j, j'$ as an \emph{ordered pair} of jobs in $\mathcal J^*$; there are separate summands for $j, j'$ and for $j', j$ and we also allow $j = j'$.
 
For each class $\mathcal P_k$ we define $\mathcal P^*_k = \mathcal P_k \cap \mathcal J^*$ and for each cluster $\mathcal C_{k,\ell}$, we define $\mathcal C^*_{k,\ell} = \mathcal C_{k,\ell} \cap \mathcal J^*$. Within each class $\mathcal P_k$,    the final opened cluster $\mathcal C^*_{k,\ell}$ after processing job $j^*$ is called the \emph{leftover cluster for $k$}, and denoted by $\cleft $.

With these conventions, we can write $Z = Z^{(i^*,j^*)}$ and $\lb = \lb^{(i^*,j^*)}$ more compactly as:
\begin{align*}
Z&= \frac{1}{2} \Bigl( \sum_{j} X_j p_{j}^2  + \sum_{j,j'} X_j X_{j'} p_j p_{j'}   \Bigr)  = \sum_{j} X_j p_{j}^2  + \frac{1}{2} \sum_{\substack{j, j': j' \neq j}} X_j X_{j'} p_j p_{j'} \\
\lb &= \frac{1}{2} \Bigl( \sum_{j} x_j p_{j}^2 + \sum_{j,j'} x_{j,j'} p_j p_{j'}  \Bigr)  \quad = \sum_{j} x_j p_{j}^2 + \frac{1}{2} \sum_{\substack{j, j': j' \neq j}}  x_{j,j'} p_j p_{j'}  
\end{align*}

Following \cite{srin1}, we further define two important quantities for measuring the approximation ratio.
$$
L = \sum_{j \in \mathcal J^*} x_j p_j, \qquad Q = \sum_{j \in \mathcal J^*} x_j p_j^2
$$

\begin{theorem}
\label{thm:partitionrel2}
For any vector $y \in [0,1]^{\mathcal J^*}$ there holds\footnote{Theorem~\ref{thm:partitionrel2} was shown in \cite{srin1} for an integral vector $y \in \{0,1 \}^{\mathcal J^*}$, but we will need the slightly generalized result for our analysis.}
$$
\lb \geq \frac{1}{2} \Bigl ( Q +  \sum_{j}  y_j   x_{j} p_j^2 + \Bigl(  L - \sum_{j} (1- \sqrt{1 - y_j}) \cdot  x_{j} p_j  \Bigr)^2 \Bigr)
$$
\end{theorem}
\begin{proof}
Define $z_j = \sqrt{1-y_j}$. We decompose the expression for $\lb$ as:
\begin{align}
\frac{1}{2} \Bigl( \sum_{j } x_j p_{j}^2 + \sum_{j , j'} x_{j,j'} p_j p_{j'} \Bigr)  =\frac{1}{2} \Bigl( \sum_j x_j p_j^2 +  \sum_{j, j'} (1 - z_j z_{j'}) x_{j,j'} p_j p_{j'} + \sum_{j,j'} x_{j,j'}  z_j z_{j'} p_j p_{j'} \Bigr) 
\label{zeqn-tt1}
\end{align}

The sum $\sum_{j, j'} (1 - z_j z_{j'}) x_{j,j'} p_j p_{j'}$ can be lower-bounded by only including the terms with $j = j'$, which contribute $\sum_{j} (1 - z_j^2) x_j p_j^2 = \sum_j y_j x_j p_j^2$.  

For the sum $\sum_{j,j'} x_{j,j'}  z_j z_{j'} p_j p_{j'}$, let $\mu = \sum_j z_j x_j p_j = L - \sum_j (1 - \sqrt{1-y_j}) x_j p_j$. Define a vector $v \in \mathbb R^{|\mathcal J| + 1}$ by setting $v_0 = -\mu, v_j = z_j p_j$ for $j  \in \mathcal J^*$, and $v_j = 0$ for $j \in \mathcal J \setminus \mathcal J^*$. The SDP relaxation ensures that $v^{\top} \xmat v \geq 0$. We calculate: 
\begin{align*}
 v^{\top} \xmat v &= \xmat_{0,0} v_0^2 + 2 \sum_{j} \xmat_{0,j} v_j v_0 + \sum_{j,j'} \xmat_{j,j'} v_j v_{j'} \\
&= \mu^2 - 2 \sum_j x_j z_j p_j \mu + \sum_{j,j'} x_{j,j'} z_j z_{j'} p_j p_{j'} = -\mu^2 + \sum_{j,j'} x_{j,j'} z_j z_{j'} p_j p_{j'}
\end{align*}
So $\sum_{j,j'} x_{j,j'} z_j z_{j'} p_j p_j' \geq \mu^2$. Substituting into Eq.~(\ref{zeqn-tt1}) gives:
\[
\lb  \geq \frac{1}{2} \Bigl(  Q + \sum_j y_j x_j p_j^2 + \Bigl( L - \sum_j (1-\sqrt{1-y_j}) x_j p_j \Bigr)^2 \Bigr).   \qedhere
\]
\end{proof}

\begin{corollary}[\cite{srin1}]
\label{thm:partitionrel2a}
There holds $\lb \geq \max\{ Q, \tfrac{1}{2}( Q + L^2) \}$.
\end{corollary}
\begin{proof}
Apply Theorem~\ref{thm:partitionrel2} with respectively vectors $y = \vec 1$ and $y = \vec 0$.
\end{proof}

At this point, we can give a simple explanation of how the algorithm achieves a $1.5$-approximation where we only assume \emph{nonpositive} correlation among the variable $X_j$. To get a better approximation ratio we need to show there is strong negative correlation among them.

\begin{proposition}[assuming nonpositive correlation only]
\label{weakzprop}
There holds $\bE[Z] \leq 3/2  \cdot \lb$.
\end{proposition}
\begin{proof}
For any jobs $j, j'$ we have $\bE[X_j] = x_j$ and $\bE[X_j X_{j'}] \leq x_j x_{j'}$. So 
\begin{align*}
\bE[Z] \leq & \sum_j p_j^2 x_j + \frac{1}{2} \sum_{\substack{j, j': j' \neq j}} p_j p_{j'} x_j x_{j'} =  \sum_j p_j^2 (x_j - x_j^2/2) + \frac{1}{2} \sum_{j',  j} p_j p_{j'} x_j x_{j'}  \leq  Q  + \frac{1}{2} L^2 
\end{align*}

On the other hand, Corollary~\ref{thm:partitionrel2a} gives \[
\lb \geq \max\{ Q, \tfrac{1}{2}( Q + L^2) \} \geq \frac{1}{3} Q + \frac{2}{3} (  \tfrac{1}{2}( Q + L^2) ) = \frac{2}{3} Q + \frac{1}{3} L^2. \qedhere 
\]
\end{proof}

\section{Determining the approximation ratio}
We now begin the arduous computation of the approximation ratio. The analysis has three main steps. First, we compute an upper bound on the expected value of the solution returned by the rounding algorithm. Second, we compute a lower bound on the objective function of the relaxation. Finally, we combine these two estimates. 

We will introduce  parameters $\kappa, \beta, \delta, \gamma$ to bound various internal functions, along with related numerical constants $c_0, \dots, c_6$. All calculations are carried out using exact arithmetic in the Mathematica computer algebra system; some specific calculation details are deferred to Appendix~\ref{num-app}. We emphasize that these parameters are not used directly in the algorithm itself.

\subsection{The algorithm upper-bound}
As a starting point, we have the following identity:

\begin{lemma}
\label{lem25}
If we condition on random variable $\roff$, then the expectation over the procedure \textsc{DepRound} satisfies 
$$
\bE[ Z \mid \roff]  \leq Q  + \frac{L^2}{2}- \sum_{k,\ell} P_k^2 B_{k,\ell}
$$
where for each cluster $k,\ell$ we define the ``bonus term''
$$
B_{k,\ell} =  \frac{1}{2}  \Bigl(  \sum_{j \in \mathcal C^*_{k,\ell}} x_j^2 H_j^2  + \sum_{\substack{j, j' \in \mathcal C^*_{k,\ell}: j \neq j'}} \phi_{j,j'}  x_j x_{j'} H_j H_{j'} \Bigr)
$$
where recall that $\phi_{j, j'} = \Phi(x_j, x_{j'}; \rho_j, \rho_{j'} ).$
\end{lemma}
\begin{proof}
All calculations in this proof are conditioned on $\roff$. We have:
$$
\bE[Z] =\bE \Bigl[ \sum_{j} p_j^2 X_j + \frac{1}{2} \sum_{j' \neq j} p_j p_{j'} X_j X_{j'} \Bigr]
$$

Here $\bE[X_j] = x_j$, and for any pair $j, j'$, we have $\bE[X_j X_{j'}] \leq x_j x_{j'}$. Also, for any cluster $k,\ell$ and distinct jobs $j, j' \in \mathcal C^*_{k, \ell}$, we have $\bE[X_j X_{j'}] \leq (1-\phi_{j,j'}) x_j x_{j'}$ by Proposition~\ref{phi-mr-prop}. So: 
\begin{align*}
&\bE[ Z ] \leq \sum_{j} p_j^2 x_j +  \frac{1}{2} \sum_{\substack{j, j': j \neq j'}} p_j p_{j'} x_j x_{j'} - \frac{1}{2} \sum_{k,\ell}  \sum_{\substack{j, j' \in \mathcal C^*_{k,\ell}: j \neq j'}} \phi_{j,j'} x_j x_{j'} p_j p_{j'} \\ 
& =  \sum_{j} p_j^2 (x_j - x_j^2/2) + \frac{1}{2} \sum_{j',  j} p_j p_{j'} x_j x_{j'} - \frac{1}{2} \sum_{k,\ell} \sum_{\substack{j, j' \in \mathcal C^*_{k,\ell}: j \neq j'}} \phi_{j,j'} x_j x_{j'} p_j p_{j'} \\
& =   Q + \frac{L^2}{2}   - \frac{1}{2} \sum_{k,\ell} \Bigl( \sum_{j \in \mathcal C^*_{k,\ell}} p_j^2 x_j^2 + \negthickspace   \negthickspace \sum_{\substack{j, j' \in \mathcal C^*_{k,\ell}: j \neq j' }} \negthickspace \negthickspace \phi_{j,j'} x_j x_{j'} p_j p_{j'} \Bigr)
\end{align*}
It remains to observe that, for any cluster $k, \ell$ and job $j \in \mathcal C^*_{k,\ell}$, we have $p_j = H_j P_k$, and so 
\begin{align*}
& \sum_{j \in \mathcal C^*_{k,\ell}} p_j^2 x_j^2 + \negthickspace   \negthickspace \sum_{\substack{j, j' \in \mathcal C^*_{k,\ell}: j \neq j' }} \negthickspace \negthickspace \phi_{j,j'} x_j x_{j'} p_j p_{j'} =  \sum_{j \in \mathcal C^*_{k,\ell}} (P_k^2 H_j^2) x_j^2 + \negthickspace   \negthickspace \sum_{\substack{j, j' \in \mathcal C^*_{k,\ell}: j \neq j' }} \negthickspace \negthickspace \phi_{j,j'} x_j x_{j'} (P_k H_j) (P_k H_{j'}) \\
&= P_k^2 \Bigl( \sum_{j \in \mathcal C^*_{k,\ell}} x_j^2 H_j^2 + \negthickspace   \negthickspace \sum_{\substack{j, j' \in \mathcal C^*_{k,\ell}: j \neq j' }} \negthickspace \negthickspace \phi_{j,j'} x_j x_{j'} H_j H_{j'} \Bigr) = 2 P_k^2 B_{k,\ell}. \qedhere \end{align*}
\end{proof}

The main focus of the analysis is to show that $B_{k,\ell}$ is large on average. For intuition, keep in mind the extremal case when all untruncated jobs have infinitesimal mass (i.e. $x_j^2 \approx 0$).  The truncated jobs require some special analysis. We also need to show that non-infinitesimal mass can only increase the value $B_{k,\ell}$.

\begin{lemma}
\label{lem26a}
For a cluster $k, \ell$, let $\mathcal U, \mathcal T$ be the set of untruncated jobs and truncated jobs respectively in $\mathcal C^*_{k,\ell}$ and define parameters  as follows:
\begin{align*}
\lambda = \sum_{j \in \mathcal C^*_{k,\ell}} \tilde \rho_j, \quad a = \frac{ \ee^{1/\lambda}-1}{\ee^{1/\lambda}+1}, \quad r = \sum_{j \in \mathcal U} x_j, \quad s = \sum_{j \in \mathcal U} x_j H_j, \quad y = \sum_{j \in \mathcal T} x_j, \quad d = \sum_{j \in \mathcal T} x_j H_j
\end{align*}
We have:
\begin{align*}
&B_{k,\ell} \geq d^2/2 + s/2 \cdot \inf_{x \in (0,r]} \Bigl(   x   + (a -  c_0 x) (s - x)  +  2 d \Phi( x, y; x/\tau, 1 - r/\tau) \Bigr)
\end{align*}
\text{for constant $c_0 := \cZeroValue$}.
\end{lemma}
\begin{proof}
We first claim that there holds:
\begin{align}
\label{lem26abnd}
\sum_{\substack{j, j' \in \mathcal U:  j \neq j'}} \phi_{j,j'}   x_j x_{j'} H_j H_{j'} \geq \sum_{j \in \mathcal U} x_j H_j (a - c_0 x_j) (s - x_j) \quad  
\end{align}

For, each untruncated job $j$ has $\rho_j = \tilde \rho_j / \lambda = x_j / \lambda$ where $\lambda \leq \tau \leq 3/4$. Because we only add jobs with non-zero mass to clusters, it also must have $x_j > 0$. By Lemma~\ref{gatt7}, any pair $j, j' \in \mathcal U$ thus has $\phi_{j,j'} \geq a - b (x_j + x_{j'})$ for $b = 0.57 \max\{0, \lambda - 0.45 \} \leq 0.57 (\tau - 0.45) = c_0/2$. So:
\begin{align*}
& \sum_{\substack{j, j' \in \mathcal U: j \neq j'}} \phi_{j,j'}   x_j x_{j'} H_j H_{j'} \geq  \sum_{ \substack{ j, j' \in \mathcal U: j \neq j'}} \bigl( a - c_0/2 \cdot (x_j + x_{j'}) \bigr) x_j x_{j'} H_j H_{j'}  \notag \\
& \qquad = \sum_{j, j' \in \mathcal U: j \neq j'} \bigl(  a - c_0 x_j \bigr) x_j x_{j'} H_j H_{j'}  \qquad \text{(by symmetry)}  \notag \\
& \qquad = \sum_{j \in \mathcal U} x_j H_j  \bigl(  a - c_0 x_j \bigr)  \bigl( -x_j H_j + \sum_{j' \in \mathcal U} x_{j'}  H_{j'}  \bigr) = \sum_{j \in \mathcal U} x_j H_j (a - c_0 x_j) (- x_j H_j + s)
\end{align*}

Finally, since $a \geq \frac{ \ee^{1/\tau} - 1}{\ee^{1/\tau} + 1} \geq 0.68$ and $H_j \geq 1$, we have $(a - c_0 x_j) (s - x_j H_j) \geq  (a - c_0 x_j) (s - x_j)$, which establishes the bound (\ref{lem26abnd}).

Now, noting that $|\mathcal T| \leq 1$, we can decompose $B_{k,\ell}$ as follows:
\begin{align*}
2 B_{k,\ell} &= \sum_{\jtrunc \in \mathcal T} x_{\jtrunc}^2 H_{\jtrunc}^2 + \sum_{j \in \mathcal U} x_j^2 H_j^2 +
\sum_{\substack{j, j' \in \mathcal U: j \neq j'}}  \phi_{j,j'} x_j x_{j'} H_j H_{j'}  + 2 \sum_{\substack{j \in \mathcal U, \jtrunc \in \mathcal T}} \phi_{j, \jtrunc} x_j x_{\jtrunc} H_j H_{\jtrunc}  \\
 &\geq \sum_{\jtrunc \in \mathcal T} x_{\jtrunc}^2 H_{\jtrunc}^2 + \sum_{j \in \mathcal U} x_j H_j \Bigl( x_j + (a - c_0 x_j) (s - x_j)  + 2 \sum_{\jtrunc \in \mathcal T} \phi_{j, \jtrunc} x_{\jtrunc} H_{\jtrunc} \Bigr) 
\end{align*}

If $\mathcal C^*_{k,\ell}$ has a truncated job $\jtrunc$ with $\rho_{\jtrunc} = 1 - r / \lmax$, then for any job $j \in \mathcal U$ we have $\phi_{j, \jtrunc} = \Phi( x_j, y; x_j/\tau,  1 - r/\tau)$ and we get:
$$
2 B_{k,\ell} \geq d^2  +  \sum_{j \in \mathcal U} x_j H_j \bigl(  x_j + (a - c_0 x_j) (s - x_j) +2 d \Phi( x_j, y; x_j/\tau,  1 - r/\tau) \bigr)
$$
and the result follows since $x_j \in (0,r]$.
\end{proof}

There is one additional complication in the analysis.  To take advantage of stochastic processing-time classes, we need a bound for the bonus term accounting for every job $j$ \emph{individually}, irrespective of how it is clustered with others. We handle this via the following additional approximation:

\begin{proposition}
\label{cor27a}
Define parameter $\kappa = \kappaValue$. For a non-leftover cluster $\mathcal C^*_{k,\ell}$, we have
$$
B_{k,\ell} \geq c_1 \sum_{j \in \mathcal C^*_{k,\ell}} x_j (H_j - \kappa) \qquad \text{for constant $c_1 := \cOneValue$}.
$$
\end{proposition}
\begin{proof}
Let $\lambda, a,r,s, y, d$ be as in Lemma~\ref{lem26a}, and define $V = \sum_{j \in \mathcal C^*_{k,\ell}} x_j (H_j - \kappa)$.

If $\mathcal C^*_{k,\ell}$ has no truncated job, then $y = d = 0$ and $\lambda= r \in [\theta,\lmax]$ and $V = s - \kappa r$. So $a = \frac{\ee^{1/r} - 1}{\ee^{1/r} + 1}$ and Lemma~\ref{lem26a} gives
$$
\frac{B_{k,\ell}}{V} \geq \inf_{\substack{r \in [\theta,\lmax], x \in (0,r] \\ s \in [r, \pi r]}} \frac{ s \bigl( x + (\frac{\ee^{1/r} - 1}{\ee^{1/r} + 1} - c_0  x)(s - x) \bigr) }{2( s - \kappa r)}
$$
and, as we show in Appendix~\ref{num-app}, this is at least $c_1$.

If $\mathcal C^*_{k,\ell}$ has a truncated job, then  $y \geq \lmax - r$, and so $V= s + d - \kappa(r + y)$ and $\lambda = \tau$ and $a = \frac{\ee^{1/\lmax}-1}{ \ee^{1/\lmax}+1}$. Lemma~\ref{lem26a} gives:
$$
\frac{B_{k,\ell}}{V} \geq \inf_{\substack{r \in [0, \theta], x \in (0,r] \\ y \in [\lmax - r,1] \\ s \in [r, \pi r], d \in [y, \pi y] }} \frac{ d^2 + s \bigl( x + (\frac{\ee^{1/\lmax}-1}{ \ee^{1/\lmax}+1}  - c_0 x) (s - x) +  2 d  \Phi( x, y; x/\tau, 1 - r/\tau) \bigr)}{ 2(s + d - \kappa (r + y))}
 $$
We show in Appendix~\ref{num-app} this is at least $0.5921116 \geq c_1$. 
\end{proof}

The non-leftover clusters are considered the ``baseline'' value. We also need to show how the leftover bonus value changes compared to this baseline.
\begin{proposition}
\label{cor26a}
For a leftover cluster $\mathcal C^*_{k,\ell} = \cleft$, define
$$
\rleft_k = \sum_{j \in \cleft} x_j, \qquad  \sleft_k = \sum_{j \in \cleft} x_j H_j, \qquad  T_k = \frac{S_k}{R_k}
$$

Then $\rleft_k \in [0,\theta], T_k \in [1,\pi]$ and $$
B_{k,\ell} \geq R_k \cdot \min\{ c_1 (T_k - \kappa), c_2 R_k T_k^2 \} \qquad \text{for constant $c_2 := \cTwoValue$}
$$

(If $R_k = 0$, we may set $T_k$ to an arbitrary value in $[1,\pi]$ and the result still holds.)
\end{proposition}
\begin{proof}
The cluster $\cleft$ has no truncated jobs and has $\rleft_k \in [0,\theta]$, as otherwise it would be closed after adding its final job. In the language of Lemma~\ref{lem26a}, we have $r = \rleft_k, s = \sleft_k = R_k T_k$, and $y = d = 0$,  and $a  \geq \frac{\ee^{1/\tau} - 1}{ \ee^{1/\tau} +1 } \approx 0.67930078 \geq 2 c_2 = 0.6793$. 

Lemma~\ref{lem26a} gives $B_{k,\ell} \geq s/2 \cdot \inf_{x \in (0,r]} ( x + (2 c_2 - c_0 x )(s - x) \bigr)$. Thus, removing a common factor of $r$ and letting $t = T_k$, it suffices to show that $$
t/2 \cdot ( x + (2 c_2 - c_0 x )(r t - x) )  \geq \min \{  c_1 (t - \kappa), c_2 r t^2 \}.
$$
 This is an algebraic inequality which can be verified to hold for $r \in [0,\theta], t \in [1, \pi], x \in [0,r]$.
\end{proof}

We thus can get the following upper bound for the algorithm performance.
\begin{lemma}
\label{algexpcor}
Define function $f$ and related random variable $D$ by:
$$
f(r,t) = \max  \{ 0,  c_1 (t - \kappa) - c_2 r t^2   \}, \qquad  \qquad D = \sum_k P_k^2 R_k f(\rleft_k, T_k), 
$$

We have $$
\bE[Z] \leq \bE[D] + c_3 Q + L^2/2 \qquad \text{for constant $c_3 := \cThreeValue$}
$$
\end{lemma}
\begin{proof}
By Proposition~\ref{cor26a}, each leftover cluster $\cleft = \mathcal C^*_{k,\ell}$  has $$
B_{k,\ell} \geq R_k \min\{ c_1 (T_k - \kappa), c_2 R_k T_k^2 \} = -R_k f(\rleft_k, T_k) + c_1 \sum_{j \in \cleft} x_j (H_j - \kappa).
$$

Putting this together with Proposition~\ref{cor26a}, we can sum over all clusters in a class $\mathcal P^*_k$ to get:
\begin{align*}
\sum_{\ell} B_{k,\ell}  &\geq \bigl( -R_k f(\rleft_k, T_k) + c_1 \sum_{j \in \cleft} x_j (H_j - \kappa) \bigr) +  \sum_{\substack{\text{non-leftover } \mathcal C^*_{k,\ell}}}  c_1 \sum_{j \in \mathcal C^*_{k,\ell}} x_j (H_j - \kappa) \\
&= -R_k f(\rleft_k, T_k) + c_1 \sum_{j \in \mathcal P^*_k} x_j (H_j - \kappa) 
\end{align*}

Summing over all classes, and noting that $P_k = p_j / H_j$ for $j \in \mathcal P^*_{k}$, we have
\begin{align*}
\sum_{k, \ell} P_k^2 B_{k,\ell} &\geq \sum_k  (- P_k^2 R_k f(\rleft_k, T_k) ) +  c_1 \sum_k P_k^2 \sum_{j \in \mathcal P^*_k } x_j (H_j - \kappa)    \\
&= -D +   c_1  \sum_k \sum_{j \in \mathcal P^*_k} x_j (H_j - \kappa)  p_j^2 / H_j^2 = -D +   c_1 \sum_{j} x_j   p_j^2 (H_j - \kappa)  / H_j^2
\end{align*}

Lemma~\ref{lem25} gives  $\bE[ Z \mid \roff ] \leq Q +  L^2/2 - \sum_{k,\ell} P_k^2 B_{k,\ell}$. So we have shown at this point that
\begin{align}
\label{ttu881}
\bE[ Z  \mid \roff] &\leq  Q  + L^2/2 + D +   c_1 \sum_{j} x_j   p_j^2 (H_j - \kappa)  / H_j^2
\end{align}

We now integrate over random variable $\roff$; by Observation~\ref{hest-prop}, for any job $j$ we have
$$
\bE[ (H_j - \kappa)  / H_j^2 ] = \frac{1}{\log \pi} \int_{h=1}^\pi \frac{h - \kappa}{h^3} \ \mathrm{d}h = \frac{\kappa - 2 \pi + 2 \pi^2 - \kappa \pi^2}{2 \pi^2 \log \pi} 
$$
So, in expectation, the term $\sum_{j} x_j   p_j^2 (H_j - \kappa)  / H_j^2$  contributes $\sum_{j} x_j   p_j^2  \cdot \frac{ \kappa - 2 \pi + 2 \pi^2 - \kappa \pi^2}{2 \pi^2 \log \pi}$, and thus $$
\bE[Z] \leq (1 - c_1 \cdot \frac{ \kappa - 2 \pi + 2 \pi^2 - \kappa \pi^2}{2 \pi^2 \log \pi}) Q + L^2/2 + D.
$$

 Direct numerical calculation shows that $1 - c_1 \cdot \frac{ \kappa - 2 \pi + 2 \pi^2 - \kappa \pi^2}{2 \pi^2 \log \pi}   \leq c_3$.
\end{proof}

\subsection{The algorithm lower-bound}
\label{lbsec}
For this section, we will choose a parameter $\beta := \betaValue$, and we define related functions $g, g_k$ by
\begin{align*}
g(r,t,h) = h \cdot \Bigl( 1 - \sqrt{1 -  \frac{\beta (h - \kappa) f(r,t)}{ h^2 (t-\kappa)}} \Bigr), \qquad \qquad
g_k(h) = g( \rleft_k, T_k, h)
\end{align*}

In analyzing functions $f(r,t)$ and $g(r,t,h)$, we implicitly assume throughout that $r \in [0,\theta]$ and $t, h \in [1, \pi]$. It can be checked that in this region $\frac{\beta (h - \kappa)  f(r,t)}{h^2 (t - \kappa)} \leq 1$, so $g$ is well-defined. Likewise, it can be checked that, for any $r,t$, the map $h \mapsto g(r,t,h)$ is non-decreasing and concave-down.

\begin{lemma}
\label{lprimelemma}
We have $$
\lb \geq \frac{1}{2} \bigl( Q + \beta D  + (L - A)^2  \bigr) \quad \text{for random variable $A = \sum_{k} P_k \sum_{j \in \cleft}x_j g_k(H_j)$}
$$
\end{lemma}
\begin{proof}
Define vector $y \in [0,1]^{\mathcal J^*}$ as follows: for a job $j \in \cleft$, we set $$
y_j =  \frac{\beta (H_j - \kappa) f (\rleft_k, T_k )  }{ H_j^2 (T_k - \kappa) } \in [0,1];
$$
note that $1 - \sqrt{1 - y_j} = g_k(H_j)/H_j$. For all other jobs $j$ we set $y_j = 0$.  By Theorem~\ref{thm:partitionrel2}, we have:
\begin{align*}
2 \cdot \lb &\geq Q +   \sum_k  \sum_{j \in \cleft} \frac{\beta f (\rleft_k, T_k )( H_j - \kappa)}{H_j^2 (T_k - \kappa)} \cdot x_j p_j^2  +  \Bigl( L -  \sum_k \sum_{j \in \cleft}  x_j p_j   \cdot  g_k(H_j)/H_j   \Bigr)^2 \\
&=  Q + \beta\sum_k P_k^2  \sum_{j \in \cleft} x_j \frac{f (\rleft_k, T_k )(H_j - \kappa)}{T_k - \kappa}  +\Bigl( L - \sum_k P_k \sum_{j \in \cleft}  g_k( H_j) x_j \Bigr)^2 \\
&= Q + \beta D  +( L - A )^2. \qedhere
\end{align*}
\end{proof}

We turn to bounding the random variable $A$. This is quite involved and requires a number of intermediate calculations.

\begin{proposition}
\label{gammaprop0}
Define parameter $\delta := \deltaValue$. We have
\begin{align*}
&A^2 \leq \sum_k P_k^2 \sum_{j \in \cleft} x_j g_k(H_j) \Bigl( c_4 + R_k \cdot \frac{T_k + \delta}{H_j + \delta} \cdot g_k(H_j) \Bigr) 
& \quad \text{for constant $c_4 := \cFourValue$.}
\end{align*}
\end{proposition}
\begin{proof}
For each class $\mathcal P_k$, define $A_k = \sum_{j \in \cleft} x_j g_k(H_j)$. We expand the sum $A^2$ as:
$$
A^2 = \bigl( \sum_k P_k A_k \bigr)^2 = \sum_k P_k A_k \bigl( P_k A_k + 2 \sum_{\ell < k} P_{\ell} A_{\ell}  \bigr) = \sum_k P_k^2  \bigl( A_k^2 + 2  A_k \sum_{\ell < k} \pi^{\ell-k }  A_{\ell}  \bigr) $$

 Since $g_{\ell} (h)$ is a concave-down function, Jensen's inequality gives
$$
A_{\ell} \leq \Bigl( \sum_{j \in \cleftl} x_j \Bigr) \cdot  g_{\ell} \Bigl(  \frac{ \sum_{j \in \cleftl} x_j H_j}{ \sum_{j \in \cleftl} x_j} \Bigr) = \rleft_{\ell} \cdot g( \rleft_{\ell}, T_{\ell}, T_{\ell})
$$

Furthermore, since $g(r,t,h)$ is an algebraic function, we can use standard algorithms to calculate
$$
c_4 \geq \max_{\substack{r \in [0,\theta], t \in [1,\pi]}} \frac{2 r g( r, t, t) }{\pi-1}
$$

So $A_{\ell} \leq c_4 \cdot \frac{\pi-1}{2}$ for each $\ell$; by the geometric series formula, we thus get:
$$
A^2 \leq P_k^2  \bigl( A_k^2 + 2 A_k  \sum_{\ell < k} \bigl( c_4 \cdot \tfrac{\pi-1}{2} \bigr) \pi^{\ell-k } \bigr) = \sum_k P_k^2 \bigl( A_k^2 + c_4 A_k \bigr)
$$
To bound the term $A_k^2$, we use the Cauchy-Schwarz inequality:
\begin{align*}
A_k^2 &= \Bigl( \sum_{j \in \cleft} \sqrt{x_j (H_j + \delta)} \cdot \frac{g_k( H_j) \sqrt{x_j}}{\sqrt{H_j+\delta}} \Bigr)^2 
\leq \Bigl( \sum_{j \in \cleft} x_j (H_j + \delta) \Bigr) \cdot \Bigl(  \sum_{j \in \cleft}\frac{g_k( H_j)^2 x_j }{H_j + \delta} \Bigr) \\
&=R_k (T_k + \delta) \sum_{j \in \cleft} \frac{x_j g_k( H_j)^2}{H_j + \delta}. \qedhere
\end{align*}
\end{proof}

\begin{lemma}
\label{gammaprop}
Define parameter $\gamma := \gammaValue$. There holds
$$
\bE[A^2] \leq \gamma \bE[ D ] + c_5 Q  \qquad \text{for constant $c_5 := \cFiveValue$}.
$$
\end{lemma}
\begin{proof}
Note that any class $\mathcal P^*_k$ has $\sum_{j \in \cleft}  x_j \cdot \frac{H_j + \delta}{T_k + \delta} = R_k$. Combining the decomposition in Proposition~\ref{gammaprop0} with the formula  $D = \sum_k P_k^2 R_k f(\rleft_k, T_k)$, we have:
\begin{align}
A^2 - \gamma D &\leq \sum_k P_k^2 \sum_{j \in \cleft} x_j \Bigl( g_k(H_j) \bigl( c_4 + \frac{R_k (T_k + \delta) g_k(H_j)}{H_j + \delta} \bigr)  - \gamma \frac{H_j + \delta}{T_k + \delta}  f(\rleft_k, T_k)  \Bigr) \label{ta5}
\end{align}

Accordingly, let us define function $F: [1,\pi] \rightarrow \mathbb R$ by 
$$
F(h) = \max_{\substack{r \in [0,\theta], t \in [1,\pi]}} \ \ \ g(r,t,h) \bigl( c_4 + \frac{ r  (t + \delta)  g(r,t,h)}{h + \delta} \bigr)-  \frac{ \gamma f(r,t) (h+\delta) }{t+ \delta}
$$

Since $\rleft_k \in [0,\theta], T_k \in [1, \pi]$,  Eq.~(\ref{ta5}) implies that
\begin{equation}
\label{ta50}
A^2 \leq \gamma D + \sum_k P_k^2  \sum_{j \in \cleft} x_j F(H_j) = \gamma D + \sum_{k} \sum_{j \in \cleft} x_j p_j^2 F(H_j) / H_j^2
\end{equation}

Note that $f(0,0) = g(0,0,h) = 0$ for all $h$, and so values $r = t = 0$ witness that $F(h) \geq 0$. We can further get an upper bound by summing over all jobs (not just the leftover jobs):
\begin{align*}
A^2 \leq \gamma D + \sum_{j \in \mathcal J^*}  x_j p_j^2  F(H_j) / H_j^2
\end{align*}

At this point, we take expectations over random variable $\roff$, getting
$$
\bE[ A^2 ] \leq \gamma \bE[D] +  \sum_{j \in \mathcal J^*}  x_j p_j^2 \bE [ F(H_j) / H_j^2  ]
$$

By Observation~\ref{hest-prop}, each item $j$ has $ \bE[ F(H_j)/H_j^2 ]  = \frac{1}{\log \pi} \int_{h=1}^{\pi} F(h)/h^3 \ \mathrm{d}h$, giving us
\begin{align*}
 \bE[ A^2 ] &\leq \gamma \bE[D] +  \Bigl( \frac{1}{\log \pi} \int_{h=1}^{\pi} \frac{F(h)}{h^3} \ \mathrm{d}h  \Bigr) Q
 \end{align*}

We show in Appendix~\ref{num-app} that $c_5 \geq \frac{1}{\log \pi} \int_{h=1}^{\pi} \frac{F(h)}{h^3} \ \mathrm{d}h$, which yields the claimed result.
\end{proof}

\subsection{Marrying the upper and lower bounds}
At this point, we have an upper bound which depends on quantities $Q, D, L$, and we have a lower bound which depends on quantities $Q, D, L, A$. We now can combine these two estimates; critically, we boil down the statistics $Q, D$ into the single parameter:
$$
q  = Q + \bE[D]/c_3
$$

\begin{observation}
\label{rest-obs}
There holds $\bE[Z] \leq c_3 q + L^2/2$
\end{observation}
\begin{proof}
Restatement of  Lemma~\ref{algexpcor}.
\end{proof}

\begin{proposition}
\label{ll0prop}
There holds $\bE[ (L-A)^2 ] \geq \max\{ 0, L - c_6 \sqrt{q} \}^2$
for constant $c_6 := \cSixValue$.
\end{proposition}
\begin{proof}
From Lemma~\ref{gammaprop}, we have $\bE[A^2] \leq \gamma \bE[ D ] + c_5 Q = q \cdot \frac{ \gamma \bE[D] + c_5 Q}{\bE[D]/c_3 + Q} \leq q \cdot \max \{ \frac{\gamma}{1/c_3}, \frac{c_5}{1} \}$; it can be checked numerically that this is at most $c_6^2 q$.  Now consider random variable $U = A^2$. The function $u \mapsto (L - \sqrt{u})^2$ is concave-up, so by Jensen's inequality,
$$
\bE[ (L-A)^2 ] = \bE[ (L - \sqrt{U})^2 ] \geq (L - \sqrt{\bE[U]})^2
$$

The bound $\bE[A^2] = \bE[U] \leq c_6^2 q$ implies that $(L - \sqrt{\bE[U]})^2 \geq \max\{0, L - c_6 \sqrt{q} \}^2$.
\end{proof}

\begin{proposition}
\label{ll1prop}
We have $$
\lb \geq  \frac{ \beta c_3  q+ \max\{0, L - c_6 \sqrt{q} \}^2}{\beta c_3 + 1}.
$$
\end{proposition}
\begin{proof}
By Corollary~\ref{thm:partitionrel2a}, we have $\lb \geq Q$. By Lemma~\ref{lprimelemma}, we have $\lb \geq \frac{1}{2} \bigl( Q + \beta D  + (L - A)^2  \bigr)$. We can take a convex combination of these two lower bounds to get:
$$
\lb \geq \alpha Q + (1-\alpha) \cdot \tfrac{1}{2} \bigl( Q + \beta D  +( L - A )^2 \bigr) \qquad \qquad \text{for $\alpha = \frac{\beta c_3 - 1}{\beta c_3 + 1} \approx 0.23$}. 
$$

Taking expectations over $\roff$ and rearranging, we get:
\begin{align*}
\lb &\geq Q \cdot \frac{1 + \alpha}{2}  + \frac{ (1 - \alpha) \beta}{2} \cdot \bE[D]  + \frac{1 - \alpha}{2} \cdot  \bE[(L - A)^2]   = \frac{ \beta c_3 q + \bE[ (L - A)^2 ]}{\beta c_3 + 1}.
\end{align*}

By Proposition~\ref{ll0prop}, we have $\bE[(L-A)^2] \geq \max\{ 0, L - c_6 \sqrt{q} \}^2$.
\end{proof}

\begin{proposition}
\label{final-lb-prop}
The approximation ratio, and the SDP integrality gap, are at most $\rhoValue$.
\end{proposition}
\begin{proof} 
By Corollary~\ref{cor:partitionrel3}, we need to bound the ratio $\bE[Z]/\lb$. Let $v = \sqrt{q}/L$. From Observation~\ref{rest-obs} and Proposition~\ref{ll1prop}, we have:
\[
\frac{ \bE[Z]}{ \lb } \leq \frac{ ( \beta c_3 + 1)   ( c_3 q+ L^2/2) }{ \beta c_3 q +  \max\{0, L -  c_6 \sqrt{  q } \}^2 } =  \frac{ ( \beta c_3 + 1)   ( c_3 v^2  + 1/2) }{ \beta c_3 v^2  + \max\{0, 1 - c_6 v \}^2 }.
\]

This is an algebraic function of $v$; its maximum value, over any $v \geq 0$, is at most $1.39798$.
\end{proof}

\section{Acknowledgments}
We thank Aravind Srinivasan for clarifications about the paper \cite{srin},  Nikhil Bansal for clarifications about the paper \cite{srin1}, and Shi Li for discussions and explanations regarding the time-indexed LP.

\appendix

 \section{Proof of Lemma~\ref{gatt7}}
 \label{gatt7app}
 We begin with a few preliminary calculations. 

 \begin{observation} 
  \label{hat553eqn}
 For $a_1 \geq a'_1 \geq 1,a_2 \geq a_2' \geq 1, b \geq 0$, there holds $$
 \frac{ (a_1-1) (a_2-1) }{ a_1 a_2 +b}  \geq \frac{ (a_1' -1) (a_2'-1)}{ a'_1 a'_2 +b}.
 $$
  \end{observation}
  \begin{proof}
Can be shown mechanically e.g  via decidability of first-order theory of real-closed fields.
\end{proof}

\begin{proposition}
\label{concave-up-prop}
For $t \geq 4/3$ and $x \in [0,1/t]$,  we have $(1-x t)^{1 - 1/x} \geq \ee^{t} (1 + x (t^2/2 - t))$.
\end{proposition}
\begin{proof}
Consider the function $f(x) = (1 - x t)^{1-1/x}$ and let $z = t^2/2 - t$. Here $e^t(1 + x z)$ is simply the first-order Taylor expansion of $f$ around $x = 0$. Thus, it suffices to show that $f''(x) \geq 0$ for $x \in [0,1/t]$. If we denote $y = 1 - t x \in [0,1]$, we calculate:
$$
f''(x) = \frac{ t^3 y^{-(1 -y + t)/(1-y)}}{(1-y)^4} \Bigl(  t ( 1 - y + \log y )^2 -  (1-y)(1 - y^2 + 2 y \log y) \Bigr)
$$
   
   Since $t \geq 4/3$, it thus suffices to show that
   $
    4/3 \cdot ( 1 - y + \log y )^2 -  (1-y)(1 - y^2 + 2 y \log y) \geq 0
    $ holds for $y \in [0,1]$. This is a one-variable inequality which can be shown via straightforward calculus. 
\end{proof}
 
We now turn to proving Lemma~\ref{gatt7}. Rephrasing Theorem~\ref{gatt5}, we want to show that:
 $$
 \frac{ ((1-x_1 t)^{1 - 1/x_1} - 1) ( (1-x_2 t)^{1 - 1/x_2} - 1)}{ (1 - x_1 t)^{1-1/x_1}  (1-x_2 t)^{1-1/x_2}+x_1 t + x_2 t-1} \geq \frac{\ee^{t} - 1}{\ee^{t} + 1} - 0.57 (x_1 + x_2) \max\{0, 1/t - 0.45 \}
 $$
 
Let $z = t^2/2 - t$. By combining Observation~\ref{hat553eqn} and Proposition~\ref{concave-up-prop}, we get
\begin{align*}
 \frac{ ( (1-x_1 t)^{1 - 1/x_1} - 1) ((1-x_2 t)^{1 - 1/x_2} - 1)}{ (1 - x_1 t)^{1-1/x_1}  (1-x_2 t)^{1-1/x_2}+x_1 t + x_2 t-1} \geq   \frac{ (\ee^{t} (1 + z x_1) - 1) (\ee^{t} (1 + z x_2) - 1)}{ \ee^{2 t} (1 + z x_1)  (1 + z x_2)+x_1 t + x_2 t -1}
\end{align*}

Let $d = x_1 + x_2, a = \frac{\ee^{t} - 1}{\ee^{t} + 1}, b = 0.57 \max\{0, 1/t - 0.45 \}$; to show Lemma~\ref{gatt7}, it thus suffices to show that
\begin{equation}
\label{ras1}
 (\ee^{t} (1 + z x_1) - 1) (\ee^{t} (1 + z x_2) - 1) - (a - b d)(  \ee^{2 t} (1 + z x_1) (1 + z x_2)+ d t -1 )
 \geq 0
 \end{equation}
 
 We can rearrange the LHS of Eq.~(\ref{ras1}) as:
 $$
 z^2 x_1 x_2 \ee^{2 t} (1 - a + b d) + b d \bigr(  \ee^{2 t} - 1 + d( t + z \ee^{2 t}) \bigr) + a d (z \ee^{t} - t) 
 $$

 Since $d \in [0,1/t]$ and $a \leq 1$, it  suffices to show:
 
 \begin{equation}
\label{ras3}
b \bigr(  \ee^{2 t} - 1 + \max\{ 0, 1 + z \ee^{2 t} / t \} \bigr) + a  \bigl( z \ee^{t} - t \bigr)  \geq 0
 \end{equation}
 where here $a,b,z$ are all functions of the single variable $t$. For $t \geq 2.22$, this is immediate since $z \ee^t - t \geq 0$. We can divide the remaining search region $[4/3, 2.22]$ into strips of width $\eps = 1/2000$. Within each strip, we use interval arithmetic in Mathematica to bound the LHS of Eq.~(\ref{ras3}). Over all regions, it is at least $0.0057 > 0$.

\section{Numerical analysis}
\label{num-app}
Recall that we have the following values for the relevant parameters:
\begin{align*}
&\pi = \piValue && c_0 = \cZeroValue \\
&\theta = \thetaValue && c_1 = \cOneValue \\
&\lmax = \lmaxValue && c_2 = \cTwoValue \\
&\beta = \betaValue && c_3 = \cThreeValue \\
&\delta = \deltaValue & & c_4 = \cFourValue\\
&\gamma = \gammaValue && c_5 = \cFiveValue \\
&\kappa = \kappaValue && c_6 = \cSixValue
\end{align*}

\subsection{Calculation of $c_1$}
Let us first calculate the value
\begin{equation}
\label{c1aeqn3}
c_1' := \min_{\substack{r \in [\theta, \lmax] \\ s \in [r, \pi r], x \in [0,r]}}  \frac{ s \cdot \bigl( x + (\frac{\ee^{1/r} - 1}{ \ee^{1/r} + 1} - c_0 x)(s - x) \bigr) }{ 2(s - \kappa r)}
 \end{equation}
Define $t = s/r$.  First, if $s \leq 4/3$, then $x + (\frac{\ee^{1/r} - 1}{\ee^{1/r} + 1} - c_0 x)(s - x)$ is an increasing function of $x$. So we can lower-bound it in this range by its value at $x = 0$, namely:
 $$
 \frac{ s^2 \cdot \frac{\ee^{1/r} - 1}{\ee^{1/r} + 1} }{ 2(s - \kappa r)} =
\frac{r (\ee^{1/r} - 1)}{2 \ee^{1/r} + 2}   \cdot  \frac{ t^2 }{ t - \kappa} \geq \frac{\theta (\ee^{1/\theta} - 1)}{2 \ee^{1/\theta} + 2} \cdot 4 \kappa \approx 0.591909465...,$$
 where the second inequality comes from straightforward calculus. 
 
 Next, we consider the case where $s \geq 4/3$. Since $r \leq \lmax$, we can estimate:
 \begin{align}
 \label{c1aeqn4}
 \frac{ s \bigl( x + (\frac{\ee^{1/r} - 1}{\ee^{1/r} + 1} - c_0 x)(s - x) \bigr) }{2( s - \kappa )r}  \geq  \frac{ s \bigl( x + (\frac{\ee^{1/\tau} - 1}{\ee^{1/\tau} + 1} - c_0 x)(s - x) \bigr) }{2(s - \kappa r)}
 \end{align}
 
 Since the RHS of Eq.~(\ref{c1aeqn4}) is an algebraic function of $r,s,x$, it can be checked automatically that its mimimum value is $0.6561...$ over the domain $r \in [\theta, \lmax], s \in [\max \{r, 4/3 \}, \pi r], x \in [0,r]$.  Putting the two cases together, we have shown that $c_1' \geq 0.591909465..$.
 
We next turn to calculating the value 
\begin{equation}
\label{c1beqn1}
c_1'' := 
\inf_{\substack{r \in [0,\theta], x \in (0,r] \\  y \in [\lmax - r,1], s \in [r, \pi r] \\ d \in [y, \pi y]}}  \frac{ d^2+ s \bigl( x + (\frac{\ee^{1/\lmax}-1}{ \ee^{1/\lmax}+1} - c_0 x) (s - x) +   2 d \Phi(x, y; x/\tau, 1 - r/\tau)}{ 2(s + d - \kappa (r + y))}
 \end{equation}
 where 
 $$
\Phi(x, y; x/\tau, 1 - r/\tau) =\frac{ ( (1-x/\lmax)^{1 - 1/x} - 1) ( (r/\lmax)^{1 - 1/y} - 1)}{ (1 -x/\lmax)^{1-1/x}  (r/\lmax)^{1-1/y} - (r-x)/\lmax}.
$$

From Proposition~\ref{concave-up-prop}, we note the following bound:
 $$
  (1 -x/\lmax)^{1-1/x} \geq  \ee^{1/\lmax} (1 - z x) \qquad \text{for $z = \lmax^{-1} - \lmax^{-2}/2 \approx 0.285$}
  $$
  
We will divide the search region for $r,y$ into boxes  $[r_{\min}, r_{\max}] \times [y_{\min}, y_{\max}]$ of width $\eps = 1/1000$. Within a given box (and subject to all other constraints on the variables), we can approximate:
  \begin{align*}
 (1 -x/\lmax)^{1-1/x}  &\geq u_0 (1 - z x)\\
(r/\lmax)^{1 - 1/y} &\geq (r_{\max}/\lmax)^{1 - 1/y_{\max}} \geq u_1 \\
(r-x)/\lmax &\geq (r_{\min} - x)/\lmax \\
 \frac{\ee^{1/\lmax}-1}{ \ee^{1/\lmax}+1} &\geq u_2
\end{align*}
for $u_0, u_1,u_2$ rational numbers within $10^{-7}$ of $\ee^{1/\lmax}, (r_{\max}/\lmax)^{1 - 1/y_{\max}},  \frac{\ee^{1/\lmax}-1}{ \ee^{1/\lmax}+1}$ respectively. By Observation~\ref{hat553eqn} and these other approximations, we can lower-bound the expression in Eq.~(\ref{c1beqn1}) by:
\begin{equation}
\label{c1beqn3}
       \frac{ d^2 + s \bigl( x + (u_2  - c_0 x) (s - x) +  \frac{2 d ( u_0 (1 - z x) - 1)(u_1 - 1)}{ u_0 (1 - z x) u_1 - (r_{\min}-x)/\lmax}  \bigr) }{ 2 (s + d - \kappa (r_{\min} + y_{\min})) }
         \end{equation}
   which is an algebraic function of $x,s,d$ and now no longer depends upon $r$ or $y$. For each such box, we use Mathematica algorithms to minimize over $s,d,x$ subject to constraints $s \in [r_{\min}, \pi r_{\max}], d \in [y_{\min}, \pi y_{\max}], x \in [0, r_{\max}]$.   With $\eps = 1/1000$, we calculate a lower bound of $c_1'' \geq 0.5921116$.

\subsection{Calculation of $c_5$}
Recall the functions $f, g, F$ defined as
\begin{eqnarray*}
&f(r,t) = \max\{0,    c_1 (t - \kappa) -   c_2 r t^2  \}, \qquad g(r,t,h) = h \cdot \Bigl( 1 - \sqrt{1 - \frac{\beta (h - \kappa)  f( r,t ) }{ h^2 (t - \kappa)}} \Bigr)  \\
&F(h) = \max_{\substack{r \in [0,\theta] \\ t \in [1,\pi]}} g(r,t,h) (c_4 + \frac{r (t + \delta) g(r,t,h)}{h + \delta}) - \frac{ \gamma (h+\delta) f(r,t) }{t + \delta}
\end{eqnarray*}

Since $g(r,t,h)$ is a non-decreasing function of $h$, we can get an upper bound within a region $h \in [h_{\min}, h_{\max}]$ by:
$$
F(h) \leq F^{\max}[h_{\min}, h_{\max}] := \max_{\substack{r  \in \mathbb R, t \in [1, \pi]}} g(r,t,h_{\max}) (c_4 + \frac{r (t + \delta) g(r,t,h_{\max}) }{h_{\min} + \delta} ) -  (h_{\min} + \delta) \frac{ \gamma f(r,t) }{t + \delta}
$$

This can be evaluated easily by Mathematica for given $h_{\min}, h_{\max}$. Note that  the constraint $r \in [0,\theta]$ has been relaxed to $r \in \mathbb R$; this is necessary for Mathematica to compute the optimization efficiently.    To handle the overall integral $ \int_{h=1}^{\pi}  F(h)/h^3 \ \mathrm{d}h$, we divide the integration region $[1,\pi]$ into strips of width $\eps = 10^{-3}$, and estimate:
\begin{align*}
&\int_{h=1}^{\pi} \frac{ F(h) \ \mathrm{d}h}{h^3} \leq \sum_{i=0}^{\lceil (\pi-1)/\eps \rceil -1} \int_{h=1 + i \eps}^{1+(i+1) \eps}  \frac{F^{\max}[1 + i \eps, 1 + (i+1)\eps]  \ \mathrm{d}h}{h^3} \\
& \qquad = \sum_{i=0}^{\lceil (\pi-1)/\eps \rceil -1} F^{\max}[1 + i \eps, 1 + (i+1)\eps]  \cdot \Bigl(\frac{1}{2 (i \eps+1)^2}-\frac{1}{2 ((i+1) \eps
  +1)^2} \Bigr)
\end{align*} 
which is calculated to be at most $c_5$.

\end{document}